\documentclass[11pt]{article}

\usepackage{amsmath}
\usepackage{amsfonts}
\usepackage{amssymb}
\usepackage{tikz}
\usepackage{amsthm}
\usepackage{fullpage}
\usepackage{caption}
\usepackage[linesnumbered,lined,boxed,commentsnumbered,ruled,vlined]{algorithm2e}
\usepackage{bm}
\usepackage[numbers]{natbib}
\usepackage[colorlinks=true,citecolor=blue, linkcolor=blue, urlcolor=blue]{hyperref}

\usepackage{comment}

\excludecomment{suppress}

\newtheorem{theorem}{Theorem}[section]

\newtheorem{lemma}[theorem]{Lemma}
\newtheorem{proposition}[theorem]{Proposition}
\newtheorem{corollary}[theorem]{Corollary}
\theoremstyle{definition}

\newtheorem{definition}{Definition}[section]
\newtheorem{remark}{Remark}[section]
\newtheorem*{remark*}{Remark}

\newcommand{\concept}[1]{\emph{#1}}

\newcommand{\LubyGlauber}{\emph{LubyGlauber}}
\newcommand{\LocalMetropolis}{\emph{LocalMetropolis}}
\newcommand{\LOCAL}{$\mathsf{LOCAL}$}
\newcommand{\DTV}[2]{d_{\mathrm{TV}}\left({#1},{#2}\right)}
\newcommand{\Diam}{{\mathrm{diam}}}
\newcommand{\dist}{{\mathrm{dist}}}
\newcommand{\LG}{\mathsf{LG}}
\newcommand{\LM}{\mathsf{LM}}

\newcommand{\QQ}[1]{\left(1-\frac{2}{q}\right)^{#1}}

\title{What can be sampled locally?
\thanks{This research is supported by the National Science Foundation of China under Grant No.~61722207 and No.~61672275.}
}
\author{
Weiming Feng~\thanks{School of Informatics, University of Edinburgh. Email: {wfeng@ed.ac.uk}.}
\and
Yuxin Sun~\thanks{Department of Computer Sciences, University of Wisconsin-Madison. Email: {yxsun@cs.wisc.edu}.}
\and
Yitong Yin~\thanks{State Key Laboratory for Novel Software Technology, Nanjing University. Email: {yinyt@nju.edu.cn}.}
}
\date{}

\begin{document}
\maketitle
\begin{abstract}
The local computation of Linial [FOCS'87] and Naor and Stockmeyer [STOC'93] studies whether a locally defined distributed computing problem is  locally solvable. In classic local computation tasks, the goal of distributed algorithms is to construct a feasible solution for some constraint satisfaction problem (CSP) locally defined on the network.

In this paper, we consider the problem of sampling a uniform CSP solution by distributed algorithms in the ~\LOCAL{} model, and ask whether a locally definable joint distribution is locally sample-able. 
We use Markov random fields (MRFs) and Gibbs distributions to model locally definable joint distributions.
We give two distributed algorithms based on Markov chains, called \LubyGlauber{} and \LocalMetropolis, which we believe to represent two basic approaches for distributed Gibbs sampling. The  algorithms achieve respective mixing times $O(\Delta\log n)$ and $O(\log n)$ under typical mixing conditions, where $n$ is the number of vertices and $\Delta$ is the maximum degree of the graph.

We show that the time bound $\Theta(\log n)$ is optimal for distributed sampling. We also show a strong $\Omega(\mathrm{diam})$ lower bound: in particular for sampling independent set in graphs with maximum degree $\Delta\ge 6$. This gives a strong separation between sampling and constructing locally checkable labelings.
\end{abstract}

\newpage
\section{Introduction}\label{sec:intro}

\paragraph{Local computation and the \LOCAL{} model.}
Locality of computation is a central theme in the theory of distributed computing. 
In the seminal works of Linial~\cite{linial1992locality}, and Naor and Stockmeyer~\cite{naor1995can}, the locality of distributed computation and the locally definable distributed computing problems are respectively captured by the \LOCAL{} model and the notion of \concept{locally checkable labeling (LCL)} problems. 
In the \LOCAL{} model~\cite{naor1995can, peleg2000distributed}, a network of $n$ processors is represented as an undirected graph, where each vertex represents a processor and each edge represents a bidirectional communication channel. Computations and communications are organized in synchronized rounds. In each round, each processor may receive a message of arbitrary size from each of its neighbors, perform an arbitrary local computation with the information collected so far, and send a message of arbitrary size to each of its neighbors. The output value for each vertex in a $t$-round protocol is determined by the local information within the $t$-neighborhood of the vertex. 
The local computation tasks are usually formulated as \emph{labeling} problems, such as the locally checkable labeling (LCL) problems introduced in~\cite{naor1995can}, in which the distributed algorithm is asked to construct a feasible solution of a constraint satisfaction problem (CSP) defined by local constraints with constant diameter in the network. 
Many problems can be expressed in this way, including various vertex/edge colorings, or local optimizations such as maximal independent set (MIS) and maximal matching. 

A classic question  for local computation  is whether a locally definable problem is locally computable.
Mathematically, this  asks
whether a feasible solution for a given local CSP can be \emph{constructed} using only local information.
There is a substantial body of research works dedicated to this question~\cite{awerbuch1989network, linial1992locality, naor1995can, kuhn2004cannot, kuhn2006price, kuhn2006complexity, barenboim2011deterministic, sarma2012distributed, fraigniaud2013towards,  barenboim2016locality, barenboim2016deterministic, chang2016exponential, fraigniaud2016local, ghaffari2016complexity, ghaffari2016improved, harris2016distributed, kuhn2016local, ghaffari2017distributed}.

\paragraph{The local sampling problem.}
Given an LCL problem which defines a local CSP on the network, aside from constructing a feasible solution of the local CSP, another interesting problem is  to \emph{sample} a \emph{uniform} random feasible solution, {e.g.}~to sample a uniform random proper coloring of the network $G$ with a given number of colors. 
More abstractly, given an instance of local CSP which, say, treats the vertices in the network $G(V,E)$ as variables,
a joint distribution of uniform random feasible solution $\boldsymbol{X}=(X_v)_{v\in V}$ is accordingly defined by these local constraints.
Our main question is \emph{whether a locally definable joint distribution can be sampled from locally}.

Intuitively, sampling could be substantially more difficult than labeling, because to sample a feasible solution is at least as difficult as to construct one, and furthermore, the marginal distribution of each random variable $X_v$ in a jointly distributed feasible solution $\boldsymbol{X}=(X_v)_{v\in V}$ may already encapsulate certain amount of non-local information about the solution space.

Retrieving such information about the solution space (as in sampling) instead of constructing one solution (as in labeling)
by distributed algorithms is especially well motivated in the context of distributed machine learning~\cite{newman2007distributed, doshi2009large, smyth2009asynchronous, yan2009parallel, gonzalez2011parallel, de2015rapidly, sa2016ensuring, Yang0Z16, XuLTZZ14}, where the data (the description of the joint distribution) is usually distributed among a large number of servers.

Besides uniform distributions, it is also natural to consider sampling from general non-uniform distributions over the solution space,
which are usually formulated as graphical models known as the weighted CSPs~\cite{cai2016nonnegative}, also known as factor graphs~\cite{mezard2009information}. 
In this model, a probability distribution called the \concept{Gibbs distribution} is defined over the space $\Omega=[q]^V$ of configurations, in such a way that
each constraint of the weighted CSP contributes a nonnegative factor in the probability measure of a configuration in $\Omega$. Due to Hammersley-Clifford's fundamental theorem~\citep[Theorem~9.3]{mezard2009information} of random fields, this model is universal for conditional independent (spatial Markovian)~\citep[Proposition~9.2]{mezard2009information} joint distributions. The conditional independence property roughly says that fixed a separator $S\subset V$ whose removal ``disconnects'' the variable sets $A$ and $B$, given any feasible configuration $X_S=\sigma_S$ over $S$, the configurations $X_A$ over $A$ and $X_B$ over $B$ are conditionally independent.

We are particularly interested in a basic class of weighted local CSPs, namely the \concept{Markov random fields (MRFs)}, where every local constraint (factor) is either a binary constraint over an edge or a unary constraint on a vertex. Specifically, given a graph $G(V,E)$ and a finite domain $[q]=\{1,2,\ldots, q\}$, the probability measure $\mu(\sigma)$ of each configuration $\sigma\in[q]^V$  under the Gibbs distribution $\mu$ is defined to be proportional to the weight:
\begin{align}
w(\sigma):=\prod_{e=uv\in E}A_e(\sigma_u,\sigma_v)\prod_{v\in V}b_v(\sigma_v),\label{eq:MRF-weight}
\end{align}
where $\{A_e\in \mathbb{R}_{\ge 0}^{q\times q}\}_{e\in E}$ are non-negative $q\times q$ symmetric matrices and $\{b_v\in \mathbb{R}_{\ge 0}^{ q}\}_{v\in V}$ are non-negative $q$-vectors, both specified by the instance of MRF. 
Examples of MRFs include combinatorial models such as independent set, vertex cover, graph coloring, and graph homomorphsm, or physical models such as hardcore gas model, Ising model, Potts model, and general spin systems.

\subsection{Our results}

We give two Markov chain based distributed algorithms for sampling from Gibbs distributions. Given any $\epsilon>0$, each algorithm returns a random output which is within total variation distance $\epsilon$ from the Gibbs distribution.
Our expositions mainly focus on MRFs, although both algorithms can be extended straightforwardly to general weighted local CSPs. 

In classic single-site Markov chains for Gibbs sampling, such as the Glauber dynamics, at each step a variable is picked at random and is updated according to its neighbors' current states.
A generic approach for parallelizing a single-site sequential Markov chain is to update a set of non-adjacent vertices in parallel at each step.
This natural idea has been considered in~\cite{gonzalez2011parallel}, also in a much broader context such as parallel job scheduling~\cite{daniel2004graph} or distributed Lov\'{a}sz local lemma~\cite{moser2010constructive, chung2014distributed}. 
For sampling from locally defined joint distributions, it is especially suitable because of the \concept{conditional independence} property of MRFs.

Our first algorithm, named \LubyGlauber, naturally parallelizes the Glauber dynamics by parallel updating vertices from independent sets generated by the ``Luby step'' in Luby's algorithm~\cite{alon1986fast,luby1986simple}. 
It is well known that Glauber dynamics achieves the mixing rate $\tau(\epsilon)=O\left(n\log\left(\frac{n}{\epsilon}\right)\right)$ under the \concept{Dobrushin's condition} for the decay of correlation~\cite{dobrushin1970prescribing,hayes2006simple}.
By a standard coupling argument, the \LubyGlauber{} algorithm achieves  a mixing rate $\tau(\epsilon)=O\left(\Delta\log\left(\frac{n}{\epsilon}\right)\right)$ under the same condition, where $\Delta$ is the maximum degree of the network.
In particular, for uniform proper $q$-colorings, this implies: 

\begin{theorem}\label{main-theorem-LubyGlauber}
If $q\ge\alpha\Delta$ for an arbitrary constant $\alpha>2$, there is an algorithm which samples a uniform proper $q$-coloring within total variation distance $\epsilon>0$ within \\$O\left(\Delta\log\left(\frac{n}{\epsilon}\right)\right)$ rounds of communications on any graph $G(V,E)$ with $n=|V|$ vertices and maximum degree $\Delta$, where $\Delta$ may be unbounded.
\end{theorem}

A barrier for this natural approach is that it will perform poorly on general graphs with large chromatic number. The situation motivates us to ask following questions:
\begin{itemize}
\item Is it possible to update all variables in $\boldsymbol{X}=(X_v)_{v\in V}$ simultaneously and still converge to the correct stationary distribution $\mu$?
\item
More concretely, is it always possible to sample almost uniform proper $q$-coloring, for a $q=O(\Delta)$, on any graphs $G(V,E)$ with $n=|V|$ vertices and  maximum degree $\Delta$, within $O(\log n)$ rounds of communications,  especially when $\Delta$ is unbounded?
\end{itemize}

Surprisingly, the answers to both questions are ``yes''. We give an algorithm, called the   \LocalMetropolis{} algorithm, achieving these goals. 
This is a bit surprising, since it seems to fully parallelize a process which is intrinsically sequential due to the massive local dependencies, especially on graphs with unbounded maximum degree.
The algorithm follows the \emph{Metropolis-Hastings} paradigm: at each step, it proposes to update all variables independently and then applies proper local filtrations to the proposals to ensure its convergence to the correct joint distribution. 
Our main discovery is that for locally defined joint distributions, the Metropolis filters are \concept{localizable}. 

The \LocalMetropolis{} algorithm always converges to the correct Gibbs distribution. The analysis of its mixing time is more involved.
In particular, for uniformly sampling proper $q$-coloring we show:

\begin{theorem}\label{main-theorem-LocalMetropolis}
If $q\ge\alpha\Delta$ for an arbitrary constant $\alpha>2+\sqrt{2}$, there is an algorithm for sampling uniform proper $q$-coloring within total variation distance $\epsilon>0$ in $O\left(\log\left(\frac{n}{\epsilon}\right)\right)$ rounds of communications on any graph $G(V,E)$ with $n=|V|$ vertices and maximum degree at most $\Delta\ge 9$, where $\Delta$ may be unbounded.
\end{theorem}


Neither of the algorithms abuses the power of the \LOCAL{} model: each message is of $O(\log n)$ bits if the domain size $q=\mathrm{poly}(n)$.

Due to the exponential correlation between variables in Gibbs distributions, the $O\left(\log\left(\frac{n}{\epsilon}\right)\right)$ time bound achieved in Theorem~\ref{main-theorem-LocalMetropolis} is optimal.

After the submission of this paper, two independent works~\cite{feng2018distributed, fischer2018simple} give the same distributed algorithm for sampling random $q$-coloring, which improves the \LocalMetropolis{} algorithm by introducing a step of laziness as distributed symmetry breaking. 
This new algorithm achieves an $O(\log n)$ mixing time under the Dobrushin's condition $q \geq (2+\delta)\Delta$. Furthermore, for graphs with sufficiently large maximum degree and girth at least 9, it achieves an $O(\log n)$ mixing time when $q \geq (\alpha^* + \delta)\Delta$, where $\alpha^* \approx 1.763$ is the positive root of equation $x = \mathrm{e}^{1/x}$. 
Another non-MCMC algorithm named distributed JVV sampler is given in~\cite{feng2018local}. For many locally definable joint distributions, this algorithm successfully samples a configuration within $\mathrm{polylog}(n)$ rounds in the \LOCAL{} model with high probability. In particular, this algorithm samples random $q$-coloring of triangle-free graphs within $O(\log^3 n)$ rounds in the \LOCAL{} model as long as $q \geq (\alpha^* + \delta)\Delta$. This non-MCMC sampling algorithm abuses the power of the \LOCAL{} model by assuming unlimited message-size and local computations.

It is a well known phenomenon that sampling may become computationally intractable 
when the model exhibits the non-uniqueness phase-transition property, e.g.~independent sets in graphs of maximum degree bounded by a $\Delta\ge 6$
\cite{sly2010computational, sly2014counting, galanis2015inapproximability, galanis2016inapproximability}.
For the same class of distributions, we show the following unconditional $\Omega(\Diam)$ lower bound for sampling in the \LOCAL{} model.

\begin{theorem}\label{main-theorem-diameter-lower-bound}
For $\Delta\ge 6$, there exist infinitely many graphs $G(V,E)$ with maximum degree $\Delta$ and diameter $\Diam(G)=|V|^{\Omega(1)}$ such that
any algorithm that samples uniform independent set in $G$ within sufficiently small constant total variation distance $\epsilon$ requires at least $\Omega(\Diam(G))$ rounds of communications,
even assuming the vertices $v\in V$ to be aware of $G$.
\end{theorem}

The lower bound is proved by a now fairly well-understood reduction from maximum cut to sampling independent sets when $\Delta\ge 6$~\cite{sly2010computational, sly2014counting, galanis2016inapproximability}. Specifically, we show that when $\Delta\ge 6$ there are infinitely many graphs $G(V,E)$ such that if one can sample a nearly uniform independent set in $G(V,E)$, then one can also sample an almost uniform maximum cut in an even cycle of size $|V|^{\Omega(1)}$, which is necessarily a global task because of the long-range correlation.

Theorem~\ref{main-theorem-diameter-lower-bound} strongly separates sampling from labeling problems for distributed computing:
\begin{itemize}
\item In the \LOCAL{} model it is trivial to construct an independent set (because $\emptyset$ is an independent set). In contrast, Theorem~\ref{main-theorem-diameter-lower-bound} says that sampling a uniform independent set is very much a global task for graphs with maximum degree $\Delta\ge 6$.
\item In the \LOCAL{} model any labeling problem would be trivial once the network structure $G$ is known to each vertex. In contrast, the sampling lower bound in Theorem~\ref{main-theorem-diameter-lower-bound} still holds even when each vertex is aware of $G$. Unlike labeling whose hardness is due to the locality of information, for sampling the hardness is solely due to the locality of randomness.
\item A breakthrough of Ghaffari, Kuhn and Maus~\cite{ghaffari2016complexity} shows that any labeling problem that can be solved sequentially with local information  admits a randomized protocol within $O(\mathrm{polylog}(n))$ rounds in the \LOCAL{} model. In contrast, for sampling we have an $\Omega(\Diam)$ randomized lower bound for graphs with $n^{\Omega(1)}$ diameter.
\end{itemize}

\subsection{Related work}
The topic of sequential MCMC (Markov chain Monte Carlo) sampling is extensively studied. The study of sampling proper $q$-colorings was initiated by the seminal works of Jerrum~\cite{jerrum1995very} and independently of Salas and Sokal~\cite{salas1997absence}. So far the best rapid mixing condition for general bounded-degree graphs is $q \ge \frac{11}{6}\Delta$ due to Vigoda~\cite{vigoda2000improved}. See~\cite{frieze2007survey} for an excellent survey.

The chromatic-scheduler-based parallelization of the Glauber dynamics chain was studied in~\cite{gonzalez2011parallel}. This parallel chain is in fact a special case of systematic scan for Glauber dynamics~\cite{dyer2006dobrushin,dyer2006systematic,hayes2006simple}, in which the variables are updated according to a fixed order. 

Empirical studies showed that sometimes an \emph{ad hoc} ``Hogwild!'' parallelization of sequential sampler might work well in practice~\cite{recht2011hogwild} and  the mixing results assuming bounded  asynchrony were given in~\cite{sa2016ensuring, johnson2013analyzing}.

A sampling algorithm based on the Lov\'{a}sz local lemma is given in~\cite{guo2016uniform}. When sampling from the hardcore model with $\lambda<\frac{1}{2\sqrt{\mathrm{e}}\Delta-1}$ on a graph of maximum degree $\Delta$, this sampling algorithm can be implemented in the \LOCAL{} model which runs in $O(\log n)$ rounds.

A problem related to the local sampling is the \concept{finitary coloring}~\cite{holroyd2014finitary}, in which a random feasible solution is sampled according to an unconstrained distribution as long as the distribution is over feasible solutions, rather than a specific distribution such as the Gibbs distribution.
Therefore, the nature of this problem is still labeling rather than sampling.

Our algorithms are Markov chains which randomly walk over the solution space. A related notion is the \concept{distributed random walks}~\cite{das2013distributed}, which walk over the network.

Our \LocalMetropolis{} chain should be distinguished from the parallel Metropolis-Hastings algorithm~\cite{calderhead2014general} or the parallel tempering~\cite{swendsen1986replica}, in which the sampling algorithms makes $N$ proposals or runs $N$ copies of the system in parallel for a suitably large $N$, in order to improve the dynamic properties of the Monte Carlo simulation.

\paragraph{Organization of the paper}
The models and preliminaries are introduced in Section~\ref{sec:model}. The \LubyGlauber{} algorithm is introduced in Section~\ref{sec:LubyGlauber}. The \LocalMetropolis{} algorithm is introduced in Section~\ref{sec:LocalMetropolis}. And the lower bounds are proved in Section~\ref{sec:lower-bound}.

\section{Models and Preliminaries}\label{sec:model}
\subsection{The \LOCAL{} model}

We assume Linial's \LOCAL{} model~\cite{naor1995can, peleg2000distributed} for distributed computation, which is as described in Section~\ref{sec:intro}. 
We further allow each node in the network $G(V,E)$ to be aware of upper bounds of $\Delta$ and $\log n$, where $n=|V|$ is the number of nodes.
This information is accessed only because the running time of the Monte Carlo algorithms may depend on them.

\subsection{Markov random field and local CSP}\label{sec:model-MRF}
The \concept{Markov random field (MRF)}, or \concept{spin system}, is a well studied stochastic model in probability theory and statistical physics.
Given a graph $G(V,E)$ and a set of \concept{spin states} $[q]=\{1,2,\ldots,q\}$ for a finite $q\ge 2$, a \concept{configuration} $\sigma\in[q]^V$ assigns each vertex one of the $q$ {spin states}. For each edge $e\in E$ there is a non-negative $q\times q$ symmetric matrix $A_e\in\mathbb{R}_{\ge 0}^{q\times q}$ associated with $e$, called the \concept{edge activity}; and for each vertex $v\in V$ there is a non-negative $q$-dimensional vector $b_v \in \mathbb{R}_{\geq 0}^q$ associated with $v$, called the \concept{vertex activity}. 
Then each {configuration} $\sigma\in[q]^V$ is assigned a weight $w(\sigma)$ which is as defined in~\eqref{eq:MRF-weight}.

This gives rise to a natural probability distribution $\mu$, called the \concept{Gibbs distribution},  over all configurations in the sample space $\Omega=[q]^V$ proportional to their weights, such that
$\mu(\sigma) = {w(\sigma)}/{Z}$
for each $\sigma\in\Omega$, where $Z=\sum_{\sigma\in\Omega}w(\sigma)$ is the normalizing factor. 
A configuration $\sigma\in\Omega$ is \concept{feasible} if $\mu(\sigma)>0$.

Several natural joint distributions can be expressed as MRFs:
\begin{itemize}
\item \textbf{Independent sets / vertex covers:} When $q=2$, all $A_e=\begin{bmatrix}1 & 1 \\ 1 & 0\end{bmatrix}$ and all $b_v=\begin{bmatrix}1  \\ 1 \end{bmatrix}$, each feasible configuration corresponds to an independent set (or vertex cover, if the other spin state indicates the set) in $G$, and the Gibbs distribution $\mu$ is the uniform distribution over independent sets (or vertex covers) in $G$. When $b_v=\begin{bmatrix}1 \\ \lambda  \end{bmatrix}$ for some parameter $\lambda>0$, this is the \concept{hardcore model} from statistical physics.

\item \textbf{Colorings and list colorings:} When every $A_e$ has $A_e(i,i)=0$ and $A_e(i,j)=1$ if $i\neq j$, and every $b_v$ is the all-1 vector, the Gibbs distribution $\mu$ becomes the uniform distribution over proper $q$-colorings of graph $G$. For \concept{list colorings}, each vertex $v\in V$ can only use the colors from its list $L_v\subseteq[q]$ of available colors. Then we can let each $b_v$ be the indicator vector for the list $L_v$ and $A_e$'s are the same as for proper $q$-colorings, so that the Gibbs distribution is the uniform distribution over proper list colorings.

\item\textbf{Physical model:} The proper $q$-coloring is a special case of the \concept{Potts model} in statistical physics, in which each $A_e$ has $A_e(i,i)=\beta$ for some parameter $\beta>0$ and $A_e(i,j)=1$ if $i\neq j$. When further $q=2$, the model becomes the \concept{Ising model}.
\end{itemize}

The model of MRF can be further generalized to allow multivariate asymmetric constraints, by which gives us the  \concept{weighted CSPs}, also known as the \concept{factor graphs}. 
In this model, we have a collection $\mathcal{C}$ of \concept{constraints} $c=(f_c, S_c)$ where each $f_c:[q]^{|S_c|}\to\mathbb{R}_{\ge0}$ is a \concept{constraint function} with \concept{scope} $S_c\subseteq V$. Each configuration $\sigma\in[q]^V$ is assigned a weight:
\[
w(\sigma) = \prod_{c=(f_c,S_c)\in\mathcal{C}}f_c(\sigma|_{S_c}),
\]
where $\sigma|_{S_c}$ represents the restriction of $\sigma$ on $S_c$. And the Gibbs distribution $\mu$ over all configurations in $\Omega=[q]^V$ is defined in the same way proportional to the weights. In particular, when $f_c$'s are Boolean-valued functions, the Gibbs distribution $\mu$ is the uniform distribution over CSP solutions.

A constraint $c=(f_c,S_c)$ is said to be \concept{local} with respect to network $G$ if the diameter of the scope $S_c$ in network $G$ is bounded by a constant. Local CSPs are expressive, for example:
\begin{itemize}
\item \textbf{Dominating sets:} They can be expressed by having a ``cover'' constraint on each inclusive neighborhood $\Gamma^+(v)$ which constrains that at least one vertex from $\Gamma^+(v)$ is chosen.
\item \textbf{Maximal independent sets (MISs):} An MIS is a dominating independent set.
\end{itemize}

Clearly, the MRF is a special class of weighted local CSPs, defined by unary and binary symmetric local constraints with respect to $G$.

\subsection{Local Sampling}
The local sampling problem is defined as follows. Let $G(V,E)$ be a network. Given an MRF defined on $G$ (or more generally a weighted CSP that is local with respect to $G$), where the specifications of the local constraints are given as private inputs to the involved processors,
for any $\epsilon>0$ upon termination each processor $v\in V$ outputs a random variable $X_v$ such that the total variation distance between the distribution $\nu$ of the random vector $X=(X_v)_{v\in V}$ and the Gibbs distribution $\mu$ is bounded as $\DTV{\mu}{\nu}\le \epsilon$, where the total variation distance between two distributions $\mu,\nu$ over $\Omega=[q]^V$ is defined as
\[
\DTV{\mu}{\nu}=\sum_{\sigma\in\Omega} \frac{1}{2}|\mu(\sigma)-\nu(\sigma)|=\max_{A\subseteq\Omega}|\mu(A)-\nu(A)|.
\]

\subsection{Mixing rate}
Our algorithms are given as Markov chains. Given an \concept{irreducible} and \concept{aperiodic} Markov chain $X^{(0)},X^{(1)},\ldots\in\Omega$, for any $\sigma\in\Omega$  let $\pi^{(t)}_{\sigma}$ denote the distribution of $X^{(t)}$ conditioning on that $X^{(0)}=\sigma$. For $\epsilon>0$ the \concept{mixing rate} $\tau(\epsilon)$ is defined as
\begin{align*}
\tau(\epsilon)=\max\limits_{\sigma\in\Omega}{\min{\left\{t: \DTV{\pi^{(t)}_\sigma}{\pi}\le\epsilon\right\}}},
\end{align*}

where $\pi$ is the \concept{stationary distribution} for the chain. For formal definitions of these notions for Markov chain, we refer to a standard textbook of the subject~\cite{levin2009markov}. Informally, irreducibility and aperiodicity guarantees that $X^{(t)}$ converges to the unique stationary distribution $\pi$ as $t\to\infty$, and the mixing rate $\tau(\epsilon)$ tells us how fast it converges.

\paragraph{Notations.}
Given a graph $G(V,E)$, we denote by $d_v=\deg(v)$ the degree of $v$ in $G$, $\Delta=\Delta_G$ the maximum degree of $G$, $\Diam=\Diam(G)$ the diameter of $G$, and $\dist(u,v)=\dist_G(u,v)$ the shortest path distance between vertices $u$ and $v$ in $G$.

We also denote by $\Gamma(v)=\{u\mid uv\in E\}$ the neighborhood of $v$, and $\Gamma^+(v)=\Gamma(v)\cup\{v\}$ the inclusive neighborhood. Finally we write $B_r(v)=\{u\mid \dist(u,v)\le r\}$ for the $r$-ball centered at $v$.

\section{The \LubyGlauber{} Algorithm}\label{sec:LubyGlauber}
In this section, we analyze a generic scheme for parallelizing Glauber dynamics, a classic sequential Markov chain for sampling from Gibbs distributions.

We assume a Markov random field (MRF) defined on the network $G(V,E)$, with edge activities $\boldsymbol{A}=\{A_e\}_{e\in E}$ and vertex activities $\boldsymbol{b}=\{b_v\}_{v\in V}$, which specifies a Gibbs distribution $\mu$ over $\Omega=[q]^V$.
The \concept{single-site heat-bath Glauber dynamics}, or simply the \concept{Glauber dynamics}, is a well known Markov chain for sampling from the Gibbs distribution $\mu$. Starting from an arbitrary initial configuration $X\in[q]^V$, at each step the chain does the followings:
\begin{itemize}
\item sample a vertex $v\in V$ uniformly at random;
\item resample the value of $X_v$ according to the marginal distribution induced by $\mu$ at vertex $v$ conditioning on the current spin states of $v$'s neighborhood.
\end{itemize}
It is well known (see~\cite{levin2009markov}) that the Glauber dynamics is a reversible Markov chain whose stationary distribution is the Gibbs distribution $\mu$.

Formally, supposed that $\sigma\in[q]^V$ is sampled from $\mu$, for any $v\in V$, $S\subseteq V$ and $\tau_S\in[q]^S$,  the \concept{marginal distribution} at vertex $v$ conditioning on  $\tau_S$, denoted as ${\mu}_v(\cdot\mid \tau_S)$, is defined as
\begin{align*}
\forall c\in[q],\quad {\mu}_v(c\mid \tau_S)=\Pr[\sigma_v=c\mid \sigma_S=\tau_S].
\end{align*}
In the Glauber dynamics, $X_v$ is resampled according to the marginal distribution $\mu_v(\cdot\mid X_{\Gamma(v)})$. Here $X_{\Gamma(v)}$ represents the current spin states of $v$'s neighborhood $\Gamma(v)$. For Markov random field, this marginal distribution can be computed as
\begin{align}\label{eq:marginal-distribution}
\forall c\in[q],\quad{\mu}_v(c\mid X_{\Gamma(v)})=\frac{b_v(c)\prod_{u\in\Gamma(v)}A_{uv}(c,X_u)}{\sum_{a\in[q]}b_v(a)\prod_{u\in\Gamma(v)}A_{uv}(a,X_u)}.
\end{align}

For example, when the MRF is the proper $q$-coloring, this is just the uniform distribution over available colors in $[q]$ which are not used by $v$'s neighbors.
For the Glauber dynamics to work, it is common to assume that the sum $\sum_{a\in[q]}b_v(a)\prod_{u\in\Gamma(v)}A_{uv}(a,X_u)$ is always positive, so that the marginal distributions are well-defined.\footnote{This property holds automatically for feasible configurations $X$ with $\mu(X)>0$, and is only needed when the Glauber dynamics is allowed to start from an infeasible configuration. For specific MRF, such as proper $q$-coloring, this property is guaranteed by the ``uniqueness condition'' $q\ge \Delta+1$.}

A generic scheme for parallelizing the Glauber dynamics is that at each step, instead of updating one vertex, the chain updates a group of ``non-interfering'' vertices in parallel, as follows:
\begin{itemize}
\item independently sample a random independent set $I$ in $G$;
\item for each $v\in I$, resample $X_v$  in parallel according to the marginal distribution $\mu_v(\cdot\mid X_{\Gamma(v)})$.
\end{itemize}
This can be seen as a relaxation of the chromatic-based scheduler~\cite{gonzalez2011parallel} and systematic scans~\cite{dyer2006systematic}.

A convenient way for generating a random independent set in a distributed fashion is the ``Luby step''  in Luby's algorithm for distributed MIS~\cite{alon1986fast,luby1986simple}: each vertex samples a uniform and independent ID from the interval $[0,1]$ (which can be discretized with $O(\log n)$ bits) and the vertices $v$ who are locally maximal among the inclusive neighborhood $\Gamma^+(v)$ are selected into the independent set~$I$.

The resulting algorithm is called \LubyGlauber{}, whose pseudocode is given in Algorithm~\ref{LubyGlauber}.

\begin{algorithm}
\SetKwInOut{Input}{Input}
\Input{Vertex $v\in V$ receives $\{A_{uv}\}_{u\in\Gamma(v)}$ and $b_v$ as input.}
initialize $X_v$ to an arbitrary value in $[q]$\;
\For{$t$ = 1 through $T$}{
    sample a real $\beta_v \in [0,1]$ uniformly and independently\;
    \If{$\beta_v > \max\{\beta_u\mid u \in \Gamma(v)\}$}{
      resample $X_v$ according to marginal distribution ${\mu}_v(\cdot\mid X_{\Gamma(v)})$\;
    }
}
\Return{$X_v$\;}
\caption{Pseudocode for vertex $v\in V$ in \LubyGlauber{} algorithm}\label{LubyGlauber}
\end{algorithm}
According to the definition of marginal distribution~\eqref{eq:marginal-distribution}, resampling $X_v$ can be done locally by exchanging neighbors' current spin states.
After $T$ iterations, where $T$ is a threshold determined for specific Markov random field, the algorithm terminates and outputs the current $\boldsymbol{X}=(X_v)_{v\in V}$.  

\begin{remark}
The \LubyGlauber{} algorithm can be easily extended to sample from weighted CSPs defined by local constraints $c=(f_c,S_c)\in\mathcal{C}$, by simply overriding the definition of neighborhood as $\Gamma(v)=\{u\neq v\mid \exists c\in\mathcal{C}, \{u,v\}\subseteq S_c\}$, thus $\Gamma(v)$ is the neighborhood of $v$ in the hypergraph where $S_c$'s are the hyperedges and now $I$ is the strongly independent set of this hypergraph.
\end{remark}

\subsection{Mixing of \LubyGlauber{}}

Let $\mu_{\LG}$ denote the distribution of $\boldsymbol{X}$ returned by the algorithm upon termination. As in the case of single-site Glauber dynamics, we assume that the marginal distribution~\eqref{eq:marginal-distribution} is always well-defined, and the single-site Glauber dynamics is irreducible among all feasible configurations. The following proposition is easy to obtain.

\begin{proposition}\label{prop:LubyGlauber-convergence}
The Markov chain \LubyGlauber{} is reversible and has stationary distribution $\mu$. Furthermore, under the above assumption, $\DTV{\mu_{\LG}}{\mu}$ converges to 0 as $T\to\infty$.
\end{proposition}

\begin{proof}
We prove this for a more general family of Markov chains, where the ``Luby step'' is replaced by an arbitrary way of independently sampling a random independent set $I$, as long as $\Pr[v\in I]>0$ for every vertex $v\in V$.

Let $\Omega=[q]^V$ and $P\in\mathbb{R}^{|\Omega|\times|\Omega|}_{\ge0}$ denote the transition matrix for the \LubyGlauber{} chain. We first show that the chain is reversible and $\mu$ is stationary. Specifically, this means to verify the \concept{detailed balance equation}:
\[
\mu(X)P(X,Y)=\mu(Y)P(Y,X),
\]
for all configurations $X,Y\in\Omega=[q]^V$. 

If both $X$ and $Y$ are infeasible, then $\mu(X)=\mu(Y)=0$ and the detailed balance equation holds trivially. If $X$ is feasible and $Y$ is not then $\mu(Y)=0$ and meanwhile since the chain never moves from a feasible configuration to an infeasible one, we have $P(X,Y)=0$ so the detailed balance equation is also satisfied.

It remains to verify the detailed balance equation when both $X$ and $Y$ are feasible.
Let $D=\{v\in V\mid X_v\ne Y_v\}$ be the set of disagreeing vertices. If $D$ is not an independent set, then $P(X,Y)=P(Y,X)=0$ and the detailed balance equation holds. 
Suppose that $D$ is an independent set. For any independent set $I\supseteq D$, we denote by $\Pr[X\rightarrow Y\mid I]$ the probability that within an iteration the chain moves from $X$ to $Y$ conditioning on $I$ being the independent set sampled in the first step. Therefore,
\begin{align*}
\frac{\Pr[X\rightarrow Y\mid I]}{\Pr[Y\rightarrow X\mid I]}=\frac{\prod_{v\in D}{b_v(Y_v)\prod_{u\in\Gamma(v)}{A_{uv}(Y_u,Y_v)}}}{\prod_{v\in D}{b_v(X_v)\prod_{u\in\Gamma(v)}{A_{uv}(X_u,X_v)}}}=\frac{\mu(Y)}{\mu(X)}.
\end{align*}
By the law of total probability,
\begin{align*}
\frac{P(X,Y)}{P(Y,X)}=\frac{\sum_{I\supseteq D}{\Pr(I)\Pr[X\rightarrow Y\mid I]}}{\sum_{I\supseteq D}{\Pr(I)\Pr[Y\rightarrow X\mid I]}}=\frac{\prod_{v\in D}{b_v(Y_v)\prod_{u\in\Gamma(v)}{A_{uv}(Y_u,Y_v)}}}{\prod_{v\in D}{b_v(X_v)\prod_{u\in\Gamma(v)}{A_{uv}(X_u,X_v)}}}=\frac{\mu(Y)}{\mu(X)}.
\end{align*}
Thus, the chain is reversible with respect to the Gibbs distribution $\mu$.

Next, observe that the chain will never move from a feasible configuration to an infeasible one. 
Moreover, due to the assumption that the marginal distribution~\eqref{eq:marginal-distribution} is always well-defined, once a vertex $v$ has been resampled, it will satisfy all local constraints. Therefore, the chain will be feasible once every vertex has been resampled. Since every vertex $v$ has positive probability $\Pr[v\in I]$ to be resampled, the chain is absorbing to feasible configurations.

It is easy to observe that every feasible configuration is aperiodic, since it has self-loop transition, i.e. $P(X,X) > 0$ for all feasible $X$. 
And any move  $X \rightarrow Y$ between feasible configurations $X,Y\in\Omega$ in the single-site Glauber dynamics with vertex $v$ being updated, can be simulated by a move in the \LubyGlauber{} chain by first sampling an independent set $I\ni v$ (which is always possible since $\Pr[v\in I]>0$) and then updating $v$ according to $X\to Y$ and meanwhile keeping all $v\in I\setminus\{v\}$ unchanged (which is always possible for feasible $X$). Provided the irreducibility of the single-site Glauber dynamics among all feasible configurations, the \LubyGlauber{} chain is also irreducible among all feasible configurations. Combining with the absorption towards feasible configurations and their aperiodicity,
due to the Markov chain convergence theorem~\cite{levin2009markov}, the total variation distance $\DTV{\mu_{\LG}}{\mu}$ converges to $0$ as $T \rightarrow \infty$.
\end{proof}

We then apply a standard coupling argument from~\cite{hayes2006simple, dyer2006dobrushin} to analyze the mixing rate of the \LubyGlauber{} chain.
The following notions are essential to the mixing of Glauber dynamics.
\begin{definition}[influence matrix]\label{Def:influence-matrix}
For  $v\in V$ and $\sigma\in[q]^V$, we write $\mu_v^\sigma=\mu_v(\cdot\mid \sigma_{\Gamma(v)})$ for the marginal distribution of the value of $v$, for configurations sampled from $\mu$ conditioning on agreeing with $\sigma$ at all neighbors of $v$.
For vertices $i, j\in V$, the \concept{influence} of $j$ on $i$ is defined as
\[
\rho_{i,j}:=\max\limits_{(\sigma,\tau)\in S_j}{d_{\text{TV}}(\mu_i^\sigma,\mu_i^\tau),}
\]
where $S_j$ denotes the set of all pairs of feasible configurations $\sigma,\tau\in[q]^V$ such that $\sigma$ and $\tau$ agree on all vertices except $j$. Let $R=(\rho_{i,j})_{i,j\in V}$ be the $n\times n$ \concept{influence matrix}.
\end{definition}
\begin{definition}[Dobrushin's condition]
Let $\boldsymbol{\alpha}$ be the \concept{total influence} {on} a vertex, defined by 
\begin{align*}
\boldsymbol{\alpha}:=\max_{i\in V}{\sum_{j\in V}{\rho_{i,j}}}.
\end{align*}
 We say that the Dobrushin's condition is satisfied if $\boldsymbol{\alpha}<1$.
\end{definition}
It is a fundamental result that the Dobrushin's condition is sufficient for the rapid mixing of Glauber dynamics~\cite{dobrushin1970prescribing,salas1997absence, hayes2006simple}, with a mixing rate of $\tau(\epsilon)=O\left(\frac{n}{1-\boldsymbol{\alpha}}\log\left(\frac{n}{\epsilon}\right)\right)$. Here we show that the \LubyGlauber{} chain is essentially a parallel speed up of the Glauber dynamics by a factor of $\Theta(\frac{n}{\Delta})$.

\begin{theorem}\label{thm:LubyGlauber-mixing-rate}
Under the same assumption as Proposition~\ref{prop:LubyGlauber-convergence}, if the total influence $\boldsymbol{\alpha}<1$, then the mixing rate of the  \LubyGlauber{} chain is $\tau(\epsilon)= O\left(\frac{\Delta}{1-\boldsymbol{\alpha}}\log\left(\frac{n}{\epsilon}\right)\right)$.
\end{theorem}
Consequently,
for any $\epsilon>0$ the \LubyGlauber{} algorithm can terminate within $O\left(\frac{\Delta}{1-\boldsymbol{\alpha}}\log\left(\frac{n}{\epsilon}\right)\right)$ rounds in the \LOCAL{} model and return an $\boldsymbol{X}\in[q]^V$ whose distribution $\mu_\LG$ is $\epsilon$-close to the Gibbs distribution $\mu$ in total variation distance.

\begin{remark}
In fact, Proposition~\ref{prop:LubyGlauber-convergence} and Theorem~\ref{thm:LubyGlauber-mixing-rate} hold for a more general family of Markov chains, where the ``Luby step'' could be  any subroutine which independently generates a random independent set $I$, as long as every vertex has positive probability to be selected into $I$.
In general, the mixing rate in Theorem~\ref{thm:LubyGlauber-mixing-rate} is in fact $\tau(\epsilon)= O\left(\frac{1}{(1-\boldsymbol{\alpha})\gamma}\log\left(\frac{n}{\epsilon}\right)\right)$ where $\gamma$ is a lower bound for the probability $\Pr[v\in I]$ for all $v\in V$.
\end{remark}

The following lemma is crucial for relating the mixing rate to the influence matrix. The lemma has been proved in various places~\cite{hayes2006simple, dyer2006dobrushin, sa2016ensuring}.

\begin{lemma}\label{boundexpectation}
Let $X$ and $Y$ be two random variables that take values over the feasible configurations in $\Omega=[q]^V$, then for any $i\in V$,
\begin{align*}
\mathop{\mathbf{E}}_{(X,Y)}\left[\DTV{\mu_i^X}{\mu_i^Y}\right]\le\sum_{k\in V}{\rho_{i,k}\Pr[X_k\ne Y_k]}.
\end{align*}
\end{lemma}

\begin{proof}
We enumerate $V$ as $V=\{1,2,\ldots, n\}$. For $0\le k\le n$, define $Z^{(k)}$ as that for each $j\in V$, $Z^{(k)}_{j}=X_j$ if $j>k$ and $Z^{(k)}_{j}=Y_j$ if $j\le k$.
In particular, $Z^{(0)}=X$ and $Z^{(n)}=Y$. 
Now, by triangle inequality,
\begin{align*}
\DTV{\mu_i^X}{\mu_i^Y} &=\DTV{\mu_i^{Z^{(0)}}}{\mu_i^{Z^{(n)}}}\le\sum\limits_{k=1}^{n}\DTV{\mu_i^{Z^{(k-1)}}}{\mu_i^{Z^{(k)}}}.
\end{align*}
Next, we note that $Z^{(k-1)}=Z^{(k)}$ if and only if $X_k=Y_k$. Therefore,
\begin{align*}
\DTV{\mu_i^X}{\mu_i^Y} 
&\le\sum\limits_{k=1}^{n}\mathbf{1}\{X_k\ne Y_k\}\DTV{\mu_i^{Z^{(k-1)}}}{\mu_i^{Z^{(k)}}}.
\end{align*}
Since $Z^{(k-1)}$ and $Z^{(k)}$ can only differ at vertex $k$, it follows that $(Z^{(k-1)},Z^{(k)})\in S_k$, and hence,
\begin{align*}
\DTV{\mu_i^X}{\mu_i^Y} 
&\le\sum\limits_{k=1}^{n}\mathbf{1}\{X_k\ne Y_k\}\max\limits_{(\sigma,\tau)\in S_k}\DTV{\mu_i^\sigma}{\mu_i^\tau}=\sum\limits_{k=1}^{n}{\rho_{i,k}\mathbf{1}\{X_k\ne Y_k\}}.
\end{align*}
By linearity of expectation,
\begin{align*}
\mathop{\mathbf{E}}_{(X,Y)}\left[\DTV{\mu_i^X}{\mu_i^Y}\right]\le \sum_{k\in V}{\rho_{i,k}\Pr[X_k\ne Y_k]}. &\qedhere
\end{align*}
\end{proof}

\begin{proof}[Proof of Theorem~\ref{thm:LubyGlauber-mixing-rate}:]
We are actually going to prove a stronger result: Denoted by $I$ the random independent set on which the resampling is executed, we write $\gamma_v=\Pr[v\in I]$ for each $v\in V$, and assume that for all $v\in V$, $\gamma_v\ge\gamma$ for some $\gamma>0$. Clearly, when $I$ is generated by the ``Luby step'', this holds for $\gamma=\frac{1}{\Delta+1}$. 
We are going to prove that $\tau(\epsilon)=O\left(\frac{1}{(1-\boldsymbol{\alpha})\gamma}\log\left(\frac{n}{\epsilon}\right)\right)$.

The proof follows the framework of Hayes~\cite{hayes2006simple}. 
We construct a coupling of the Markov chain $(X^{(t)}, Y^{(t)})$ such that the transition rules for $X^{(t)}\to X^{(t+1)}$ and $Y^{(t)}\to Y^{(t+1)}$ are the same as the \LubyGlauber{} chain. If $\Pr[X^{(T)}\neq Y^{(T)}\mid X^{(0)}=\sigma \wedge Y^{(0)}=\tau]\le\epsilon$ for any initial configurations $\sigma,\tau\in\Omega$, then by the coupling lemma for Markov chain~\cite{levin2009markov}, we have the mixing rate $\tau(\epsilon)\le T$. 

The coupling we are going to use is the maximal one-step coupling of the \LubyGlauber{} chain, which for every vertex $i\in V$ achieves that
\[
\Pr\left[X^{(t+1)}_i\ne Y^{(t+1)}_i\mid X^{(t)}, Y^{(t)}\right]
=\DTV{\mu_i^{X^{(t)}}}{\mu_i^{Y^{(t)}}},
\]
where $\mu_i^{X^{(t)}}$ and $\mu_i^{Y^{(t)}}$ are the marginal distributions as defined in Definition~\ref{Def:influence-matrix}. The existence of such coupling is guaranteed by the  coupling lemma.

Arbitrarily fix $\sigma,\tau\in\Omega=[q]^V$. For $t\ge 0$, define $(X^{(t)},Y^{(t)})\in\Omega^2$ by iterating a maximal one-step coupling of the \LubyGlauber{} chain, starting from initial condition $X^{(0)}=\sigma,Y^{(0)}=\tau$. Due to the well-defined-ness of marginal distribution~\eqref{eq:marginal-distribution}, we know that once all vertices have been resampled, the configuration will be feasible and will remain to be feasible in future.

Let $T_1$ be a positive integer and $\mathcal{F}$ denote the event all vertices have been resampled in chain $X$ and $Y$ in the first $T_1$ steps. By union bound, we have
\begin{align}
\Pr\left[\neg{\mathcal{F}}\right]\le2\sum\limits_{v\in V}{(1-\gamma_v)^{T_1}}\le2n(1-\gamma)^{T_1},\label{eq:T1}
\end{align}
Next, we assume that $X^{(t)},Y^{(t)}$ are both feasible for $t\ge T_1$. We define the vector $\mathbf{p}^{(t)}\in[0,1]^V$ as
\begin{align*}
\forall j\in V, \quad p^{(t)}_j:=\Pr\left[X^{(t)}_j\ne Y^{(t)}_j\right].
\end{align*}
By the definition of the \LubyGlauber{} chain, it holds for every $j\in V$ that
\begin{align}
p^{(t+1)}_j=(1-\gamma_j)p^{(t)}_j+\gamma_j\cdot \Pr\left[X^{(t+1)}_j\ne Y^{(t+1)}_j\mid j\in I\right].\label{eq:trans-LubyGlauber}
\end{align}
By the definition of maximal one-step coupling and Lemma~\ref{boundexpectation}, for $t\ge T_1$, for any $i\in V$,

\begin{align*}
&\Pr\left[X^{(t+1)}_i\ne Y^{(t+1)}_i\mid i\in I\right]\\
 =\, &\sum_{\sigma,\tau\in\Omega\atop \mu(\sigma),\mu(\tau)>0} \bigg( \Pr\left[X^{(t+1)}_i\ne Y^{(t+1)}_i\mid X^{(t)}=\sigma, Y^{(t)}=\tau\right]\cdot\Pr\left[X^{(t)}=\sigma\wedge Y^{(t)}=\tau\right]\bigg)\\
=\, & \sum\limits_{\sigma,\tau\in\Omega\atop\mu(\sigma),\mu(\tau)>0}{\DTV{\mu_i^{\sigma}}{\mu_i^{\tau}}\cdot\Pr\left[X^{(t)}=\sigma\wedge Y^{(t)}=\tau\right]}\\
=\, & \mathbf{E}\left[\DTV{\mu_i^{X^{(t)}}}{\mu_i^{Y^{(t)}}}\right] \le  \sum\limits_{k\in V}{\rho_{i,k}\cdot \Pr\left[X^{(t)}_k\ne Y^{(t)}_k\right]}.
\end{align*}

Combined with equality~\eqref{eq:trans-LubyGlauber}, for $t\ge T_1$ we have
\begin{align*}
\mathbf{p}^{(t+1)}\le M\mathbf{p}^{(t)},
\end{align*}
where matrix $M=(J-\Gamma)J+\Gamma R$, where $\Gamma$ is the $n\times n$ diagonal matrix with $\Gamma_{i,i}=\gamma_i$; $J$ is the $n\times n$ identity matrix; and $R=(\rho_{ij})$ is the influence matrix.
The $\infty$-norm of $M$ is bounded as
\begin{align*}
\lVert M\rVert_{\infty}&=\max\limits_{i\in V}{\sum\limits_{j\in V}{|M_{i,j}|}}\le\max\limits_{i\in V}{\left\{1-(1-\boldsymbol{\alpha})\gamma_i\right\}} \le1-(1-\boldsymbol{\alpha})\gamma.
\end{align*}
Let $T=T_1+T_2$. By induction, we obtain the component-wise inequality
\begin{align*}
\mathbf{p}^{(T)}\le M^{T_2}\mathbf{p}^{(T_1)}.
\end{align*}
Conditioning on that $X^{(T_1)}$ and $Y^{(T_1)}$ are both feasible, we have
\begin{align}
\Pr\left[X^{(T)}\ne Y^{(T)}\right]&\le\lVert\mathbf{p}^{(T)}\rVert_1 \qquad\text{by union bound}\notag\\
&\le n\lVert\mathbf{p}^{(T)}\rVert_{\infty} \quad\text{by H\"{o}lder's  inequality}\notag\\
&\le n\lVert M^{T_2}\mathbf{p}^{(T_1)}\rVert_{\infty}\notag\\
&\le n\lVert M\rVert_{\infty}^{T_2}\lVert\mathbf{p}^{(T_1)}\rVert_{\infty} \notag\\
&\le n\left(1-(1-\boldsymbol{\alpha})\gamma\right)^{T_2}\label{eq:T2}
\end{align}

For any $\epsilon$, we choose $T_1=\left\lceil\frac{1}{\gamma}\ln\left(\frac{4n}{\epsilon}\right)\right\rceil$ and $T_2=\left\lceil\frac{1}{(1-\boldsymbol{\alpha})\gamma}\ln\left(\frac{2n}{\epsilon}\right)\right\rceil$. Then $T=T_1+T_2=O\left(\frac{1}{(1-\boldsymbol{\alpha})\gamma}\log\left(\frac{n}{\epsilon}\right)\right)$. Combining~\eqref{eq:T1} and~\eqref{eq:T2}, conditioning on $X^{(0)}=\sigma\wedge Y^{(0)}=\tau$ for arbitrary $\sigma,\tau\in\Omega$, we have
\begin{align*}
\Pr\left[X^{(T)}\ne Y^{(T)}\right]
\le\Pr[\neg\mathcal{F}]+\Pr\left[X^{(T)}\ne Y^{(T)}\mid{\mathcal{F}}\right] \le2n(1-\gamma)^{T_1}+n\left(1-(1-\boldsymbol{\alpha})\gamma\right)^{T_2}\le\epsilon.
\end{align*}
This implies that
\begin{align*}
\tau(\epsilon)= O\left(\frac{1}{(1-\boldsymbol{\alpha})\gamma}\log\left(\frac{n}{\epsilon}\right)\right).
\end{align*}
In particular, if the random independent set $I$ is generated by the ``Luby step", we have $\gamma=\frac{1}{\Delta+1}$, therefore for the \LubyGlauber{} chain
\begin{align*}
\tau(\epsilon)= O\left(\frac{\Delta}{1-\boldsymbol{\alpha}}\log\left(\frac{n}{\epsilon}\right)\right). &\qedhere
\end{align*}
\end{proof}

\subsection{Application of \LubyGlauber{} for sampling graph colorings}

For uniformly distributed proper $q$-coloring of graph $G$, it is well known that the Dobrushin's condition is satisfied when $q\ge 2\Delta+1$ where $\Delta$ is the maximum degree of graph $G$. 

We consider a more generalized problem, the list colorings, where each vertex $v\in V$ maintains a list $L_v\subseteq[q]$ of colors that it can use. The proper $q$-coloring is a special case of list coloring when everyone's list is precisely $[q]$.
For each vertex $v\in V$, we denote by $q_v=|L_v|$ the size of $v$'s list, and $d_v=\deg(v)$ the degree of $v$. It is easy to verify that the total influence $\boldsymbol{\alpha}$ is now bounded as
\begin{align*}
\boldsymbol{\alpha}
&=\max\limits_{i\in V}{\sum\limits_{j\in V}{\rho_{i,j}}} 
=\max\limits_{v\in V}{\left\{\frac{d_v}{q_v-d_v}\right\}}.
\end{align*}

Applying Theorem~\ref{thm:LubyGlauber-mixing-rate}, we have the following corollary, which also implies Theorem~\ref{main-theorem-LubyGlauber}.
\begin{corollary}\label{col:list-coloring}

If there is an arbitrary constant $\delta>0$ such that $q_v\ge(2+\delta)d_v$ for every vertex $v$, then the mixing rate of the \LubyGlauber{} chain for sampling list coloring is
$\tau(\epsilon)= O\left(\Delta\log\left(\frac{n}{\epsilon}\right)\right)$.
\end{corollary}

\section{The \LocalMetropolis{} Algorithm}\label{sec:LocalMetropolis}

In this section, we give an algorithm that may fully parallelize the sequential process under suitable mixing conditions, even on graphs with unbounded degree. 
The algorithm is inspired by the famous Metropolis-Hastings algorithm for MCMC, in which a random choice is \concept{proposed} and then \concept{filtered} to enforce the target stationary distribution. Our algorithm, called the \LocalMetropolis{} algorithm, makes each vertex propose independently, and localizes the work of filtering to each edge. 

We are given a Markov random field (MRF) defined on the network $G(V,E)$, with edge activities $\boldsymbol{A}=\{A_e\}_{e\in E}$ and vertex activities $\boldsymbol{b}=\{b_v\}_{v\in V}$, whose Gibbs distribution is $\mu$.
Starting from an arbitrary configuration $X\in[q]^V$, in each iteration, the \LocalMetropolis{} chain does the followings:
\begin{itemize}
\item
\textbf{Propose:}
Each vertex $v \in V$ independently proposes a spin state $\sigma_v\in[q]$ with  probability proportional to $b_v(\sigma_v)$.
\item
\textbf{Local filter:}
Each edge $e\in E$ flips a biased coin independently, with the probability of HEADS being
\[
\tilde{A}_e(\sigma_u,\sigma_v)\tilde{A}_e(X_u,\sigma_v)\tilde{A}_e(\sigma_u,X_v),
\] 
where $\tilde{A}_e$ is the matrix obtained by normalizing $A_e$ as $\tilde{A}_e=A_e/\max_{i,j}A_e(i,j)$.
We say that the edge \concept{passes the check} if the outcome of coin flipping is HEADS.

Then for each vertex $v \in V$, if all edges incident with $v$ passed their checks, $v$ accepts the proposal and updates the value as $X_v=\sigma_v$, otherwise $v$  leaves $X_v$ unchanged.
\end{itemize}
After $T$ iterations, where $T$ is a threshold determined for specific Markov random field, the algorithm terminates and outputs the current $\boldsymbol{X}=(X_v)_{v\in V}$. 
The pseudocode for the \LocalMetropolis{} algorithm is given in Algorithm~\ref{LocalMetropolis}.

\begin{algorithm}
\SetKwInOut{Input}{Input}
\Input{Each vertex $v\in V$ receives $\{A_{uv}\}_{u\in\Gamma(v)}$ and $b_v$ as input.}
each $v\in V$ initializes $X_v$ to an arbitrary value in $[q]$\;
\For{$t = 1$ through $T$}{
  \ForEach{$v \in V$}{
    propose a random $\sigma_v\in[q]$ with probability $b_v(\sigma_v)/\sum_{c\in[q]}b_v(c)$;
  }
  
  \ForEach{$e = (u,v) \in E$}{
  pass the check  independently with probability 
  $\frac{{A}_e(\sigma_u,\sigma_v){A}_e(X_u,\sigma_v){A}_e(\sigma_u,X_v)}{\left(\max_{i,j\in[q]}A_e(i,j)\right)^3}$\;  
  }
  \ForEach{$v \in V$}{
    \If { all edges $e$ incident with $v$ pass the checks}{
      $X_v \gets \sigma_v$\;
    } 
  }
}
each $v\in V$ \textbf{returns}{ $X_v$\;}
\caption{Pseudocode for the \LocalMetropolis{} algorithm}\label{LocalMetropolis}
\end{algorithm}

We remark that in each iteration, for each edge $e=uv$, the two endpoints $u$ and $v$ access the same random coin to determine whether $e$ passes the check in this iteration.  

\begin{remark}
The \LocalMetropolis{} algorithm can be naturally extended to sample from weighted CSPs. 
The local filtering now occurs on each local constraint, such that a $k$-ary constraint $c=(f_c,S_c)\in\mathcal{C}$ passes the check with the probability which is a product of $2^k-1$ normalized factors $\tilde{f}_c(\tau)$ for the  $\tau\in[q]^{S_c}$ obtained from $2^k-1$ ways of mixing $\sigma_{S_c}$ with $X_{S_c}$ except the $X_{S_c}$ itself.
\end{remark}

\subsection{Mixing of \LocalMetropolis}
Let $\mu_\LM$ denote the distribution of $\boldsymbol{X}=(X_v)_{v\in V}$ returned by the \LocalMetropolis{} algorithm after $T$ iterations.

We need to ensure the chain is well behaved even when starting from infeasible configurations. Now we make the following assumption: for all $X \in [q]^V$ and $v \in V$, 
\begin{align}
\label{assumption-localmetropolis}  
\sum_{i\in[q]}b_v(i)\prod_{u \in \Gamma(v)}A_{uv}(i,X_u)\sum_{j \in [q]}b_u(j)A_{uv}(X_v,j)A_{uv}(i,j) > 0,
\end{align}
which is slightly stronger than the assumption made for the Glauber dynamics. As in the case of Glauber dynamics, the property is needed only when the chain is allowed to start from an infeasible configuration $X\in[q]^V$ with $\mu(X)=0$. 
For specific MRF, such as graph colorings, the condition~\eqref{assumption-localmetropolis} is satisfied as long as $q \geq \Delta + 1$ and $q\ge 3$.
As before, we further assume that the single-site Markov chain\footnote{For the MRFs, since the single-site Glauber dynamics has the same connectivity structure as the natural single-site version of Metropolis chain, we do not distinguish between them when referring to irreducibility.} is irreducible among feasible configurations.

\begin{theorem}
\label{theorem-localmetropolis-mixing}
The Markov chain \LocalMetropolis{} is reversible and has stationary distribution $\mu$. Furthermore, under above assumptions, $\DTV{\mu_{\LM}}{\mu}$ converges to $0$ as $T \rightarrow \infty$.
\end{theorem}
\begin{proof}
Let $\Omega=[q]^V$ and  $P \in \mathbb{R}_{\geq 0} ^{|\Omega| \times |\Omega|}$ denote the transition matrix for the \LocalMetropolis{} chain. First, we show this chain is reversible and $\mu$ is stationary, by verifying the detailed balance equation:
\begin{align*}
\mu(X)P(X,Y)=\mu(Y)P(Y,X).  
\end{align*}
If two configurations $X,Y$ are both infeasible, then \\$\mu(X)=\mu(Y)=0$. If precisely one of $X, Y$ is feasible, say $X$ is feasible and $Y$ is not, then $\mu(Y)=0$ and $X$ cannot move to $Y$ since at least one edge cannot pass its check, which means $P(X,Y) = 0$. In both cases, the detailed balance equation holds. 

Next, we suppose $X,Y$ are both feasible.
Consider a move in the \LocalMetropolis{} chain. Let $\mathcal{C}\in\{0,1\}^E$ be a Boolean vector that $\mathcal{C}_e$ indicates whether edge $e\in E$ passes its check.
We call $v\in V$ \concept{non-restricted} by $\mathcal{C}$ if $\mathcal{C}_e=1$ for all $e$ incident with $v$ and $v$ accepts the proposal; and call $v\in V$ \concept{restricted} by $\mathcal{C}$ if otherwise.

A move in the chain is completely determined by $\mathcal{C}$ along with the proposed configurations $\sigma\in[q]^V$.
Let $\Omega_{X \rightarrow Y}$ denote the set of pairs $(\sigma, \mathcal{C})$ with which $X$ moves to $Y$, and $\Delta_{X,Y}=\{v\in V\mid X_v\neq Y_v\}$ the set of vertices on which $X$ and $Y$ disagree. 
Note that each $(\sigma, \mathcal{C})\in\Omega_{X \rightarrow Y}$  satisfies:
\begin{itemize}
\item
  $\forall v \in \Delta_{X,Y}$: $\sigma_v=Y_v$ and $v$ is non-restricted by  $\mathcal{C}$;
\item
  $\forall v \not\in \Delta_{X,Y}$: either $\sigma_v  = X_v=Y_v$ or $v$ is restricted by $\mathcal{C}$.
\end{itemize} 
Similar holds for $\Omega_{Y \rightarrow X}$, the set of $(\sigma, \mathcal{C})$ with which $Y$ moves to $X$.
Hence:
\begin{align}
\label{localmetropolis-detailed-balance-equation}
\frac{{P}(X,Y)}{{P}(Y,X)}=\frac{\sum_{(\sigma,\mathcal{C})\in \Omega_{X \rightarrow Y}}{\Pr}(\sigma) {\Pr}(\mathcal{C}\mid\sigma,X)}{\sum_{(\sigma,\mathcal{C})\in \Omega_{Y \rightarrow X}}{\Pr}(\sigma) {\Pr}(\mathcal{C}\mid \sigma,Y)}.
\end{align}
In order to verify the detailed balance equation, we construct a bijection $\phi_{X,Y} : \Omega_{X \rightarrow Y} \rightarrow \Omega_{Y \rightarrow X}$, and for every $(\sigma,\mathcal{C}) \in \Omega_{X \rightarrow Y}$, denoted $(\sigma^\prime,\mathcal{C}^\prime)=\phi_{X,Y}(\sigma,\mathcal{C})$, and show that
\begin{align}
\label{bijection-equation}
\frac{{\Pr}(\sigma){\Pr}(\mathcal{C}\mid\sigma,X)}{{\Pr}(\sigma^\prime){\Pr}(\mathcal{C}^\prime\mid\sigma^\prime,Y)}=\frac{\mu(Y)}{\mu(X)}.
\end{align}
The detailed balance equation then follows from~\eqref{localmetropolis-detailed-balance-equation} and~\eqref{bijection-equation}.

The bijection $(\sigma,\mathcal{C})\stackrel{\phi_{X,Y}}{\longmapsto}(\sigma^\prime,\mathcal{C}^\prime)$ is constructed as follow:
\begin{itemize}
\item
  $\mathcal{C}^\prime = \mathcal{C}$;
\item
  for all $v$ non-restricted by $\mathcal{C}$, since $(\sigma,\mathcal{C})\in\Omega_{X \rightarrow Y}$ it must hold $\sigma_v=Y_v$, then set $\sigma'_v=X_v$;
\item
  for all $v$ restricted by $\mathcal{C}$, since $(\sigma,\mathcal{C})\in\Omega_{X \rightarrow Y}$ it must hold $X_v=Y_v$, then set $\sigma'_v=\sigma_v$.
\end{itemize}
It can be verified that the $\phi_{X,Y}$ constructed in this way is indeed a bijection from $\Omega_{X \rightarrow Y}$ to $\Omega_{Y \rightarrow X}$.
For any $(\sigma,\mathcal{C})\in \Omega_{X \rightarrow Y}$ and the corresponding $(\sigma',\mathcal{C}')\in \Omega_{Y \rightarrow X}$, since $\mathcal{C}'=\mathcal{C}$, in the following we will not specify whether $v$ is (non-)restricted by $\mathcal{C}$ or $\mathcal{C}'$ but just say $v$ is (non-)restricted, and the followings are satisfied:
\begin{itemize}
\item
  $\forall v\in \Delta_{X,Y}$:  $\sigma_v = Y_v$, $\sigma'_v=X_v$ and $v$ is non-retricted;
\item
  $\forall v \not\in \Delta_{X,Y}$: either $\sigma_v=\sigma_v'=X_v=Y_v$ or $v$ is restricted and $\sigma_v=\sigma_v'$. In both cases, $\sigma_v = \sigma'_v$.
\end{itemize}
Then we have:
\begin{align}
\label{localmetropolis-reversible-vertex}
\frac{{\Pr}(\sigma)}{{\Pr}(\sigma^\prime)}
=\frac{\prod_{v \in V}b_v(\sigma_v)}{\prod_{v \in V}b_v(\sigma'_v)}
=\frac{\prod_{v: X_v \not = Y_v}b_v(\sigma_v)}{\prod_{v: X_v \not = Y_v}b_v(\sigma'_v)}=\frac{\prod_{v: X_v \not = Y_v}b_v(Y_v)}{\prod_{v: X_v \not = Y_v}b_v(X_v)}
=\frac{\prod_{v \in V}b_v(Y_v)}{\prod_{v \in V}b_v(X_v)}.
\end{align}
Next, for each edge $e\in E$ we calculate the ratio $\frac{{\Pr}(\mathcal{C}_e\mid\sigma,X)}{{\Pr}(\mathcal{C}_e^\prime\mid\sigma^\prime,Y)}$. 
There are two cases:
\begin{itemize}
\item
  If $\mathcal{C}_e=0$ which means $e$ does not pass its check, then
\begin{align*}
  &{\Pr}[\mathcal{C}_e=0 \mid \sigma,X]
                       = 1-\tilde{A}_e(\sigma_u,\sigma_v)\tilde{A}_e(X_u,\sigma_v)\tilde{A}_e(\sigma_u,X_v)\\
&\text{and}\quad\\
&{\Pr}[\mathcal{C'}_e=0\mid\sigma',Y] = 1-\tilde{A}_e(\sigma'_u,\sigma'_v)\tilde{A}_e(Y_u,\sigma'_v)\tilde{A}_e(\sigma'_u,Y_v).
  \end{align*}
  And both $u$ and $v$ are restricted by $\mathcal{C}$. By our construction of the bijection $\phi_{X,Y}$, we have
  $\sigma_u=\sigma'_u$, $\sigma_v=\sigma'_v$, $X_u=Y_u$, and $X_v=Y_v$.
It follows that 
\begin{align*}
\frac{{\Pr}[\mathcal{C}_e=0\mid\sigma,X]}{{\Pr}[\mathcal{C^\prime}_e=0\mid\sigma^\prime,Y] } 
  =\frac{A_e(Y_u,Y_v)}{A_e(X_u,X_v)}=1.
\end{align*}
\item
If $\mathcal{C}_e=1$ which means $e$ passes its check, then
  \begin{align*}
  &{\Pr}[\mathcal{C}_e=1\mid\sigma,X]
                       = \tilde{A}_e(\sigma_u,\sigma_v)\tilde{A}_e(X_u,\sigma_v)\tilde{A}_e(\sigma_u,X_v),\\
  &\text{and }\\    &{\Pr}[\mathcal{C'}_e=1\mid\sigma',Y]                      = \tilde{A}_e(\sigma'_u,\sigma'_v)\tilde{A}_e(Y_u,\sigma'_v)\tilde{A}_e(\sigma'_u,Y_v).
  \end{align*}
  There are three sub-cases according to whether vertices $u$ and $v$ are restricted:
  \begin{enumerate}
  \item
    Both $u$ and $v$ are restricted, in which case $\sigma_u=\sigma'_u$, $\sigma_v=\sigma'_v$, $X_u=Y_u$, $X_v=Y_v$.
  \item
    Precisely one of $\{u,v\}$ is restricted, say $v$ is restricted and $u$ is non-restricted, in which case 
    $\sigma_u=Y_u$, $\sigma'_u=X_u$, $\sigma_v=\sigma'_v$, and $X_v=Y_v$.
  \item
    Both $u$ and $v$ are  non-restricted, in which case $\sigma_u = Y_u$, $\sigma'_u=X_u$, $\sigma_v=Y_v$, $\sigma'_v=X_v$.   
  \end{enumerate}
  In all three sub-cases, the following identity can be verified:
  \[
  \frac{{\Pr}[\mathcal{C}_e=1\mid\sigma,X]}{{\Pr}[\mathcal{C^\prime}_e=1\mid\sigma^\prime,Y] }
      =\frac{\tilde{A}_e(Y_u,Y_v)}{\tilde{A}_e(X_u,X_v)}=\frac{A_e(Y_u,Y_v)}{A_e(X_u,X_v)}.
  \]
\end{itemize}
Since each edges passes its check independently, we have
\begin{align}
\label{localmetropolis-reversible-edge}
\frac{{\Pr}(\mathcal{C}\mid\sigma,X)}{{\Pr}(\mathcal{C^\prime}\mid\sigma^\prime,Y) }
      =\prod_{e=uv\in E}\frac{A_e(Y_u,Y_v)}{A_e(X_u,X_v)}.
\end{align}
Combining \eqref{localmetropolis-reversible-vertex} and \eqref{localmetropolis-reversible-edge}, for every $(\sigma,\mathcal{C})\in \Omega_{X \rightarrow Y}$ and the corresponding $(\sigma',\mathcal{C}')\in \Omega_{Y \rightarrow X}$, we have:
\begin{align*}
\frac{{\Pr}[\sigma]{\Pr}[\mathcal{C}\mid\sigma,X]}{{\Pr}[\sigma^\prime]{\Pr}[\mathcal{C}^\prime\mid\sigma^\prime,Y]}=\prod_{v \in V}\frac{b_v(Y_v)}{b_v(X_v)}\prod_{e=uv\in E}\frac{A_e(Y_u,Y_v)}{A_e(X_u,X_v)}=\frac{\mu(Y)}{\mu(X)}.
\end{align*}
This completes the verification of detailed balance equation and the proof of the reversibility of the chain with respect to stationary distribution $\mu$.

Next, observe that the chain will never move from a feasible configuration to an infeasible one since at least one of the edge will not pass its check. 
By assumption~\eqref{assumption-localmetropolis}, for all $X \in [q]^V$, no matter feasible or not, and for every $v \in V$ there must be a spin state $i\in[q]$ such that with positive probability $v$ is successfully updated to spin state $i$. Note that once a vertex is successfully updated it satisfies and will keep satisfying all its local constraints. 
Therefore, the chain is absorbing to feasible configurations.

It is easy to observe that every feasible configuration is aperiodic, since it has self-loop transition, i.e. $P(X,X) > 0$ for all feasible $X$. In addition, any move  $X \rightarrow Y$ between feasible configurations $X,Y\in\Omega$  in the single-site Markov chain with vertex $v$ being updated, can be simulated by a move in the \LocalMetropolis{} chain in which all the vertices $u$ other than $v$ propose their current spin state $X_u$ and $v$ proposes $Y_v$. Provided the irreducibility of the single-site Markov chain among all feasible configurations, the \LocalMetropolis{} chain is also irreducible among all feasible configurations. Combinining with the absorption towards feasible configurations and their aperiodicity,
due to the Markov chain convergence theorem~\cite{levin2009markov}, $\DTV{\mu_{\LM}}{\mu}$ converges to $0$ as $T \rightarrow \infty$.
\end{proof}

\subsection{The mixing of \LocalMetropolis{} chain for graph colorings}
Unlike the \LubyGlauber{} chain, whose mixing rate is essentially due to the analysis of systematic scans. The mixing rate of \LocalMetropolis{} chain is much more complicated to analyze.
Here we analyze the mixing rate of the \LocalMetropolis{} chain for proper $q$-colorings. 

Given a graph $G(V,E)$, a $q$-coloring $\sigma\in[q]^V$ is proper if $\sigma_u\neq\sigma_v$ for all $uv\in E$. For this special MRF, the \LocalMetropolis{} chain behaves simply as follows. Starting from an arbitrary coloring $X\in[q]^V$, not necessarily proper,
in each step:
\begin{itemize}
\item
\textbf{Propose:}
each vertex $v$ proposes a color $c_v\in[q]$ uniformly at random;
\item
\textbf{Local filter:}
each vertex $v$ rejects its proposal if there is a neighbor $u \in \Gamma(v)$ such that one of the followings occurs:
\begin{enumerate} 
\item ($v$ proposed the neighbor's current color) $c_v  = X_u$;
\item ($v$ and the neighbor proposed the same color) $c_v  = c_u$;
\item (the neighbor proposed $v$'s current color) $X_v  = c_u$;
\end{enumerate}
otherwise, $v$ accepts its proposal and updates its color $X_v$ to $c_v$.
\end{itemize}

The first two filtering rules are sufficient to guarantee that the chain will never move to a ``less proper'' coloring.
Although at first glance the third filtering rule looks redundant, it is necessary to guarantee the reversibility of the chain as well as the uniform stationary distribution.

It can be verified that when $q \geq \Delta + 2$, the condition~\eqref{assumption-localmetropolis} is satisfied and the single-site Glauber dynamics for proper $q$-coloring is irreducible,
and hence the chain is mixing 
due to Theorem~\ref{theorem-localmetropolis-mixing}. 
The following theorem states a condition in the form $q\ge\alpha\Delta$ for the logarithmic mixing rate even for unbounded $\Delta$ and $q$. This proves Theorem~\ref{main-theorem-LocalMetropolis}.

\begin{theorem}
\label{theorem-localmetropolis-graphcoloring}
If $q\ge\alpha\Delta$ for a constant $\alpha>2+\sqrt{2}$,
the mixing rate of the \LocalMetropolis{} chain for proper $q$-coloring on graphs with maximum degree at most $\Delta=\Delta(n)\ge 9$ is $\tau(\epsilon)=O(\log\left(\frac{n}{\epsilon}\right))$, where the constant factor in $O(\cdot)$ depends only on $\alpha$ but not on the maximum degree $\Delta$.
\end{theorem}

The theorem is proved by path coupling, a powerful engineering tool for coupling Markov chains. A coupling of a Markov chain on space $\Omega$ is a Markov chain $(X,Y)\to(X',Y')$ on space $\Omega^2$ such that the transitions $X\to X'$ and $Y\to Y'$ individually follow the same transition rule as the original chain on $\Omega$.
For path coupling, we can construct a coupled Markov chain $(X,Y)\to(X',Y')$ for $X,Y\in[q]^V$ which differ at only one vertex. The chain mixes rapidly if the expected number of disagreeing vertices in $(X',Y')$ is less than 1.

\subsubsection{An ideal coupling}\label{section-ideal-coupling}
The $2+\sqrt{2}$ threshold in Theorem~\ref{theorem-localmetropolis-graphcoloring} is due to an ideal coupling in the $\Delta$-regular tree. Let $\mathbb{T}_{\Delta}$ denote the infinite $\Delta$-regular tree rooted at $v_0$. We assume that the current pair of colorings $(X,Y)$ disagree only at the root $v_0$ and $X_u=Y_u\not\in\{X_{v_0},Y_{v_0}\}$  for all other vertices $u$ in $\mathbb{T}_{\Delta}$. 

An ideal coupling can be constructed as follows in a breadth-first fashion: (1) the root $v_0$ proposes the same random color in both chains $X,Y$; (2) each child $u$ of the root proposes the same random color in both chains unless it proposed one of $\{X_{v_0},Y_{v_0}\}$, in which case it switches the roles of the two colors $\{X_{v_0},Y_{v_0}\}$ in the $Y$ chain; (3) for all other vertices $u$, it proposes the same random color in both chains unless its parent proposed different colors in the two chains, in which case $u$ switches the roles of $\{X_{v_0},Y_{v_0}\}$ in the $Y$ chain. 
For this ideal coupling, by a calculation, it can be verified that for the root $v_0$:
\[
\Pr[X_{v_0}'\neq Y_{v_0}']\le 1-\left(1-\frac{\Delta}{q}\right)\left(1-\frac{2}{q}\right)^\Delta
\]
and for any non-root vertex $u$ in $\mathbb{T}_{\Delta}$ at distance $\ell$ from $v_0$:
\begin{align*}
\Pr[X_{u}'\neq Y_{u}']\le \frac{1}{q}\left(1-\frac{2}{q}\right)^{\Delta-1}\left(\frac{2}{q}\right)^{\ell-1}=\frac{1}{2}\left(1-\frac{2}{q}\right)^{\Delta-1}\left(\frac{2}{q}\right)^{\ell}.
\end{align*}
The expected number of disagreeing vertices in $(X',Y')$ is then bounded as 
\begin{align*}
\Pr[X_{v_0}'\neq Y_{v_0}']+\sum_{u\in T\atop u\neq v_0}\Pr[X_u'\neq Y'_u]
\le& 1-\left(1-\frac{\Delta}{q}\right)\left(1-\frac{2}{q}\right)^\Delta + \frac{1}{2}\left(1-\frac{2}{q}\right)^{\Delta-1}\sum_{\ell=1}^{\infty}\Delta^\ell \left(\frac{2}{q}\right)^{\ell}\\
=&
1-\left(1-\frac{\Delta}{q}\right)\left(1-\frac{2}{q}\right)^\Delta
+\frac{\Delta}{q-2\Delta}\left(1-\frac{2}{q}\right)^{\Delta-1}.
\end{align*}
The path coupling argument requires this quantity to be less than 1. 
For $q=\alpha^{\star}\Delta$ and $\Delta\to\infty$, this quantity becomes $1-\mathrm{e}^{-2/\alpha^{\star}}\left(1-\frac{1}{\alpha^{\star}}-\frac{1}{\alpha^{\star}-2}\right)$, which is less than $1$ if $\alpha^{\star}>2+\sqrt{2}$.

For general non-tree graphs $G(V,E)$ and arbitrary pairs of colorings $(X,Y)$ which disagree at only one vertex, where $X,Y$ may not even be proper, we essentially show that the above special pair of colorings $(X,Y)$ on the infinite $\Delta$-regular tree $\mathbb{T}_{\Delta}$ represent the worst case for path coupling. The analysis for this general case is quite involved. We first state the path coupling lemma with general metric.

\begin{lemma}[Bubley and Dyer~\cite{bubley1997path}]
\label{pathcoupling}
Given a pre-metric, which is a connected undirected graph on configuration space $\Omega$ with positive edge weight such that every edge is a shortest path, let $\Phi(X,Y)$ be the length of the shortest path between two configurations $X,Y\in\Omega$. Suppose that there is a coupling $(X,Y) \rightarrow (X',Y')$ of the Markov chain defined only for the pair $(X,Y)$ of configurations that are adjacent in the pre-metric, which satisfies that
\begin{align*}
\mathbf{E}[\Phi(X',Y') \mid X,Y] \leq (1-\delta)\Phi(X,Y),  
\end{align*}
for some $0 < \delta < 1$. Then the mixing rate of the Markov chain is bounded by
\begin{align*}
\tau(\epsilon)\le\frac{\log(\Diam(\Omega)/\epsilon)}{\delta},
\end{align*}
where $\Diam(\Omega)$ denotes the diameter of $\Omega$ in the pre-metric.
\end{lemma}

We use the following slightly modified pre-metric:
A pair $(X,Y)\in\Omega=[q]^V$ is connected by an edge in the pre-metric if and only if $X$ and $Y$ differ at only one vertex, say $v$, and the edge-weight is given by $\deg(v)$. This leads us to the following definition.

\begin{definition}

For any $X', Y'\in\Omega$, for $u\in V$, we define $\phi_u(X',Y')=\deg(u)$ if $X'_u\ne Y'_u$ and $\phi_u(X',Y')=0$ if otherwise; and for $S\subseteq V$, we define the distance between $X'$ and $Y'$ on $S$ as
\begin{align*}
\Phi_S(X',Y'):=\sum\limits_{u\in S:X'_u\ne Y'_u}{\phi_u(X',Y')}.
\end{align*}
In addition, we denote $\Phi(X',Y')=\Phi_V(X',Y')$.
\end{definition}
Clearly, the diameter of $\Omega$ in distance $\Phi$ has $\Diam(\Omega)\le n\Delta$.

We prove the mixing rate in Theorem~\ref{theorem-localmetropolis-graphcoloring} for two separate regimes for $q$ by using two different couplings.
We define $\alpha^*\approx3.634\ldots$ to be the positive root of $\alpha=2\mathrm{e}^{1/\alpha}+1$.

\begin{lemma}
\label{lemma-coloring-local-coupling}
If $q \ge \alpha\Delta +3$ for a constant $\alpha>\alpha^*$, then $\tau(\epsilon)=O(\log\left(\frac{n}{\epsilon}\right))$.
\end{lemma}

\begin{lemma}
\label{lemma-coloring-global-coupling}
If $\alpha\Delta\le q\le 3.7\Delta +3$ for $2+\sqrt{2}<\alpha\le3.7$ and $\Delta\ge 9$, then $\tau(\epsilon)=O(\log\left(\frac{n}{\epsilon}\right))$.
\end{lemma}
Theorem~\ref{theorem-localmetropolis-graphcoloring} follows by combining the two lemmas.

\subsubsection{An easy local coupling for $q > 3.634\Delta +3$}
We first prove Lemma~\ref{lemma-coloring-local-coupling} by constructing a local coupling where the disagreement will not percolate outside its neighborhood.
Let $X,Y\in[q]^V$ two $q$-colorings, not necessarily proper. Assume that $X$ and $Y$ disagree only at vertex  $v_0\in V$.
The coupling $(X,Y)\to(X',Y')$ is constructed as follows:
\begin{itemize}
\item Each vertex $v\in V$ proposes the same random color in the two chains $X$ and $Y$. 
Then $(X',Y')$ is determined due to the transition rule of  \LocalMetropolis{} chain.  
\end{itemize}
Next we show the path coupling condition:
\[
\mathbf{E}[\Phi(X',Y') \mid X,Y] \le (1-\delta)\Phi(X,Y)=(1-\delta)\deg(v_0).
\]

The following technical lemma is frequently applied in the analysis of this and next couplings.
\begin{lemma}
\label{lemma-inequality}
If $q \geq a \Delta$, then for any integer $0\le d\le \Delta$, $d\left(1-\frac{a}{q}\right)^d\le \Delta\left(1-\frac{a}{q}\right)^{\Delta}$. 
\end{lemma}
\begin{proof}
It is sufficient to show the function $d\left(1-\frac{a}{q}\right)^d$ is monotone for integer $1\le d\le \Delta$:
\begin{align*}
d\left(1-\frac{a}{q}\right)^d - (d-1)\left(1-\frac{a}{q}\right)^{d-1}=\left(1-\frac{a}{q}\right)^{d-1}\left(1-\frac{ad}{q}\right),
\end{align*}
which is nonnegative when $q\ge ad$.
\end{proof}

\paragraph{Proof of Lemma \ref{lemma-coloring-local-coupling}.}
First, observe that if $v \not \in \Gamma^+(v_0)$, where $v_0$ is the vertex at which $X$ and $Y$ disagree, then it always holds that $X'_v = Y'_v$, because all vertices in $\Gamma^+(v)$ are colored the same in $X$ and $Y$ and will propose the same random color in the two chains due to the coupling. Therefore, it is sufficient to consider the difference between $X'$ and $Y'$ in  $\Gamma^+(v_0)$ and we have
\[
\Phi(X',Y')=\Phi_{\Gamma^+(v_0)}(X',Y').
\]

For each $v$, let $c_v\in[q]$ be the uniform random color proposed independently by $v$, which is identical in both chains by the coupling.

For the disagreeing vertex $v_0$, it holds that $X'_{v_0} =Y'_{v_0}$ if $v_0$ accepts the proposal in both chains, which occurs when $c_{v_0} \not \in \{X_u, Y_u: {u\in\Gamma(v_0)}\}$ and $\forall u \in \Gamma(v_0), c_u \not \in \{X_{v_0},Y_{v_0}, c_{v_0}\}$. 
Since $X$ and $Y$ disagree only at $v_0$, we have
\begin{align}
\Pr[X'_{v_0}  = Y'_{v_0}\mid X,Y]\geq \left(1-\frac{d_{v_0}}{q}\right)\left(1-\frac{3}{q}\right)^{d_{v_0}}.\label{eq:local-coupling-center}
\end{align}

For each $u \in \Gamma(v_0)$, since $X_{u}=Y_{u}$, the event $X'_u\neq Y'_u$ occurs only when  $c_u \in \{X_{v_0}, Y_{v_0}\}$ and $\forall w \in \Gamma(u)$, $c_w \not \in \{X_u, c_u\}$. Note that to guarantee $X_u'\neq Y'_u$ one must have $c_u\neq X_u$, thus

\begin{align}
\forall u\in\Gamma(v_0): \quad \Pr[X'_u  \not = Y'_u\mid X,Y]  & \leq \frac{2}{q}\left(1-\frac{2}{q}\right)^{d_u}.\label{eq:local-coupling-neighbor}
\end{align}

Combining \eqref{eq:local-coupling-center} and \eqref{eq:local-coupling-neighbor} together and due to linearity of expectation, we have
\begin{align*}
\mathbf{E}[\Phi(X',Y') \mid X,Y] =&\sum\limits_{u\in V}{\mathbf{E}[\phi_u(X',Y')\mid X,Y]}
=\sum_{u \in \Gamma^+(v_0)} d_u \Pr[X'_u \not = Y'_u \mid X,Y]
\\\leq& d_{v_0} \left[1- \left(1-\frac{d_{v_0}}{q}\right)\left(1-\frac{3}{q}\right)^{d_{v_0}}\right] + \frac{2}{q}\sum_{u \in \Gamma({v_0})}d_u\left(1-\frac{2}{q}\right)^{d_u}
\\\le& d_{v_0}\left[1-\left(1-\frac{\Delta}{q}\right)\left(1-\frac{3}{q}\right)^\Delta+\frac{2\Delta}{q}\left(1-\frac{2}{q}\right)^\Delta\right], 
\end{align*}
where the last inequality is due to the monotonicity stated in Lemma~\ref{lemma-inequality}.

The path coupling condition
is satisfied when
\begin{align}
\label{ieq-3.6}
\left(1-\frac{\Delta}{q}\right)\left(1-\frac{3}{q}\right)^\Delta-\frac{2\Delta}{q}\left(1-\frac{2}{q}\right)^\Delta \ge \delta.
\end{align}

For $q=\alpha^*\Delta$ and $\Delta\to\infty$,  then the LHS becomes $\left(1-\frac{1}{\alpha^*}\right)\mathrm{e}^{-{3}/{\alpha^*}}-\frac{2}{\alpha}\mathrm{e}^{-{2}/{\alpha^*}}$, which is 0 when $\alpha^*$ is the positive root of $\alpha^* = 2\mathrm{e}^{{1}/{\alpha^*}} + 1$.

Furthermore, for $\Delta\ge 1$ and $q\ge\alpha\Delta+3$, the LHS become:
\begin{align*}
\left(1-\frac{3}{q}\right)^\Delta\left[1-\frac{\Delta}{q}-\frac{2\Delta}{q}\left(1+\frac{1}{q-3}\right)^\Delta\right]
\ge&\left(1-\frac{3}{\alpha \Delta+3}\right)^\Delta\left[1-\frac{1}{\alpha}-\frac{2}{\alpha}\left(1+\frac{1}{\alpha\Delta}\right)^\Delta\right]
\\\ge& \frac{\mathrm{e}^{-3/\alpha}}{\alpha}(\alpha-2\mathrm{e}^{1/\alpha}-1),
\end{align*}
which is a positive constant independent of $\Delta$ when $\alpha>\alpha^*$.

Therefore, when $\alpha>\alpha^*$, there is a constant $\delta>0$ which depends only on $\alpha$, such that for all $\Delta\ge 1$ and $q\ge\alpha\Delta+3$, the inequality~\eqref{ieq-3.6} is satisfied, which by Lemma~\ref{pathcoupling},  gives us $\tau(\epsilon) = O\left(\log\left(\frac{n}{\epsilon}\right)\right)$.

\subsubsection{A global coupling for $(2+\sqrt{2})\Delta<q\le3.7\Delta+3$}
Next, we prove Lemma~\ref{lemma-coloring-global-coupling} and bound the mixing rate when $(2+\sqrt{2})\Delta<q\le3.7\Delta+3$. This is done by a global coupling where the disagreement may percolate to the entire graph, whose construction and analysis is substantially more sophisticated than the previous local coupling. Although this sophistication only improves the threshold for $q$ in Lemma~\ref{lemma-coloring-local-coupling} by a small constant factor, the effort is worthwhile because it helps us to approache the threshold  of the ideal coupling discussed in Section~\ref{section-ideal-coupling} and shows that the infinite $\Delta$-regular tree $\mathbb{T}_{\Delta}$ represents the worst case for path coupling. And curiously, the extremity of this worst case only holds when $q$ is also properly upper bounded, say~$q\le 3.7\Delta+3$, whereas the mixing rate for larger $q$ was guaranteed by Lemma~\ref{lemma-coloring-local-coupling}.

Let $v_0 \in V$ be a vertex and $X,Y \in [q]^V$ any two $q$-colorings (not necessarily proper) which disagree only at $v_0$. 
The coupling $(X,Y)\rightarrow(X',Y')$ of the \LocalMetropolis{} chain is constructed by coupling $(\boldsymbol{c}^X,\boldsymbol{c}^Y)$, where $\boldsymbol{c}^X,\boldsymbol{c}^Y\in[q]^V$ are the respective vector of proposed colors in the two chains $X$ and $Y$. For each $v\in V$, the $(c_v^X,c_v^Y)$ is sampled from one of the two following joint distributions:
\begin{itemize}
\item \textbf{consistent:} $c_v^X=c_v^Y$ and is uniformly distributed over $[q]$;
\item \textbf{permuted:} $c_v^X$ is uniform in $[q]$ and $c_v^Y=\phi(c_v^X)$ where $\phi:[q]\to[q]$ is a  bijection defined as that $\phi(X_{v_0})=Y_{v_0}$, $\phi(Y_{v_0})=X_{v_0}$, and $\phi(x)=x$ for all $x\not\in\{X_{v_0},Y_{v_0}\}$.
\end{itemize}
Note that for all $u\neq v_0$ we have $X_u=Y_u$, and if further $X_u\in\{X_{v_0},Y_{v_0}\}$, we say the vertices $w\in\Gamma^+(u)\setminus\{v_0\}$ are \concept{blocked} by $u$, and all other $u\neq v_0$ is \concept{unblocked}. The special vertex $v_0$ is neither blocked nor unblocked. 
We denote by $\Gamma^B(v)$ and $\Gamma^U(v)$ the respective sets of blocked and unblocked neighbors of vertex $v$ and let $b_v = |\Gamma^B(v)|$.

The coupling $(\boldsymbol{c}^X,\boldsymbol{c}^Y)$ of proposed colors is constructed by the following recursive procedure:
\begin{itemize}
\item Initially, for the disagreeing vertex $v_0$, $(c_{v_0}^X,c_{v_0}^Y)$ is sampled consistently in the two chains. 
\item For each unblocked $u\in \Gamma(v_0)$,  the $(c_{u}^X,c_{u}^Y)$ is sampled independently (of other vertices) from the permuted distribution. 
\item Let $\mathcal{S}\subseteq V$ denote the current set of vertices $v$ such that $(c_v^X,c_v^Y)$ has been sampled, and $\mathcal{S}^{\neq}\subseteq\mathcal{S}$ the set of vertices $v$ with $(c_v^X,c_v^Y)$ sampled inconsistently as $c_v^X\neq c_v^Y$. We abuse the notation and use $\partial \mathcal{S}^{\neq}=\{\text{unblocked }u\not\in\mathcal{S}\mid \exists uv\in E, \text{ s.t. }v\in\mathcal{S}^{\neq} \}$ to denote the unblocked un-sampled vertex boundary of $\mathcal{S}^{\neq}$. 
If such $\partial \mathcal{S}^{\neq}$ is non-empty, then all $u\in \partial \mathcal{S}^{\neq}$ sample the respective  $(c_{u}^X,c_{u}^Y)$ independently from the permuted distribution and join the $\mathcal{S}$ simultaneously. Grow $\mathcal{S}^{\neq}$  according to the results of sampling. Repeat this step until the current $\partial \mathcal{S}^{\neq}$ is empty and thus $\mathcal{S}$ is stabilized.
\item For all remaining vertices $v$, $(c_v^X,c_v^Y)$ is sampled independently and consistently.
\end{itemize}
This procedure is in fact a Galton-Watson branching process starting from root $v_0$. The blocked-ness of each vertex is determined by the current $X$ and $Y$. The $\mathcal{S}$ grows from the root by a percolation of disagreement $c_v^X\neq c_v^Y$ added in a breadth-first order. 

It is easy to see that each individual $c_v^X$ or $c_v^Y$ is uniformly distributed over $[q]$ and is independent of $c_u^X$ or $c_u^Y$ for all other $u\neq v$ (although the joint distributions $(c_v^X,c_v^Y)$ may be dependent of each other). Therefore, the $(\boldsymbol{c}^X,\boldsymbol{c}^Y)$ is a valid coupling of proposed colors. 

A walk $\mathcal{P}=(v_0,v_1,\ldots,v_\ell)$ in $G(V,E)$ is called a \concept{strongly self-avoiding walk (SSAW)} if $\mathcal{P}$ is a simple path in $G$ and $v_iv_j$ is not an edge in $G$ for any $0 < i+1<j \le \ell$. An SSAW $\mathcal{P}=(v_0,v_1,\ldots,v_\ell)$ is said to be a \concept{path of disagreement} with respect to $(\boldsymbol{c}^X,\boldsymbol{c}^Y)$ if $(c_{v_i}^X,c_{v_i}^Y), v_i\in\mathcal{P}$ are sampled in the order along the path $\mathcal{P}$ from $i=0$ to $\ell$, and $c_{v_i}^X\neq c_{v_i}^Y$ for all $1\le i\le \ell$.
For any specific SSAW $\mathcal{P}=(v_0,v_1,\ldots,v_\ell)$ through unblocked vertices $v_1,v_2,\ldots,v_{\ell}$,
by the chain rule
\begin{align}
&\Pr[\,\mathcal{P}\text{ is a path of disagreement }]\le \prod_{i=1}^\ell \Pr\left[\,c_{v_i}^X\in \{X_{v_0},Y_{v_0}\}\,\right]=\left(\frac{2}{q}\right)^\ell.\label{eq:disagreemnet-path-probability}
\end{align}

\begin{proposition}\label{prop:SSAW}
For any vertex $u\neq v_0$, the event $c_u^X\neq c_u^Y$ occurs only if there is a strongly self-avoiding walk (SSAW) $\mathcal{P}=(v_0,v_1,\ldots,v_\ell)$ from $v_0$ to $v_\ell=u$ through unblocked vertices $v_1,v_2,\ldots,v_{\ell}$ such that $\mathcal{P}$ is a path of disagreement.
\end{proposition}
\begin{proof}
By the coupling, $c^X_u \neq c^Y_u$ only when $(c^X_u, c^Y_u)$ is sampled from the permuted distribution and it must hold that $\{c^X_u,c^Y_u\} = \{X_{v_0},Y_{v_0}\}$. This means that $u$ itself must be unblocked.

At the time when $(c^X_u, c^Y_u)$ is being sampled, there must exist a neighbor $w \in \Gamma(u)$ such that either (1) $w=v_0$ or (2) $w \in \mathcal{S}^{\neq}$, which means that $c^X_w\neq c^Y_w$, $\{c^X_w,c^Y_w\} = \{X_{v_0},Y_{v_0}\}$ was sampled before $(c^X_u, c^Y_u)$, and vertex $w$ is unblocked. 
If it is the latter case, we repeat this argument for $w$ recursively until $v_0$ is reached. This will give us a path $\mathcal{P}=(v_0,v_1,\ldots,v_{\ell})$ from $v_0$ to $u=v_{\ell}$ through unblocked vertices $v_1,\ldots,v_{\ell}$ such that for all $1\le i\le \ell$, $(c^X_{v_i}, c^Y_{v_i})$ are sampled in that order, $c^X_{v_i} \neq c^Y_{v_i}$ and $\{c^X_{v_i},c^Y_{v_i}\} = \{X_{v_0},Y_{v_0}\}$.
Thus, $\mathcal{P}$ is a path of disagreement through unblocked vertices. 
Note that this path $\mathcal{P}=(v_0,v_1,\ldots,v_{\ell})$ must be a strongly self-avoiding. To the contrary assume that $\mathcal{P}$ is not strongly self-avoiding and there exist $0\le i,j\le \ell$ such that $i<j-1$ and $v_{i}v_{j}$ is an edge. In this case, right after $c^X_{v_i}\neq c^Y_{v_i}$ being sampled and $v_i$ joining $\mathcal{S}^{\neq}$,  $v_{i+1}$ and $v_{j}$ must be both in $\partial \mathcal{S}^{\neq}$ because they are both unblocked un-sampled neighbors of $v_i$ then. And due to our construction of coupling, the $(c^X_{v_{i+1}}, c^Y_{v_{i+1}})$ and $(c^X_{v_{j}}, c^Y_{v_{j}})$ are sampled and $v_{i+1}, v_{j}$ join $\mathcal{S}$ simultaneously, which contradict that $(c^X_{v_{j}}, c^Y_{v_{j}})$ is sampled after $(c^X_{v_{i+1}}, c^Y_{v_{i+1}})$ along the path.
Therefore, $\mathcal{P}$ is an SSAW through unblocked vertices and is also a path of disagreement.
\end{proof}

The coupled next step $(X',Y')$ is determined by the current $(X,Y)$ and the coupled proposed colors $(\boldsymbol{c}^X,\boldsymbol{c}^Y)$.

\begin{proposition}\label{prop:global-coupling-differ}
 For any  vertex $u\neq v_0$, the event $X'_u\neq Y'_u$ occurs only if $c^X_u,c^Y_u \in \{X_{v_0},Y_{v_0}\}$. Furthermore, for any unblocked vertex $u\neq v_0$, the event $X_u'\neq Y_u'$ occurs only if $c_u^X\neq c_u^Y$.
 \end{proposition}
 \begin{proof}
We pick any  $u\neq v_0$. Assume by contradiction that $c^X_u=c^Y_u \not \in \{X_{v_0},Y_{v_0}\}$. Note that this covers all possible contradicting cases to that $c^X_u,c^Y_u \in \{X_{v_0},Y_{v_0}\}$, because $c^X_u\neq c^Y_u$ occurs only when $c^X_u,c^Y_u \in \{X_{v_0},Y_{v_0}\}$.

We then show for every edge $uw$ incident to $u$, the followings hold:
\begin{align}
	&c^X_u = c^X_w \text{ if and only if } c^Y_u = c^Y_w, \label{eq-condition1}\\
	&X_u = c^X_w \text{ if and only if } Y_u = c^Y_w,\label{eq-condition2}\\
	&c^X_u = X_w \text{ if and only if } c^Y_u = Y_w.\label{eq-condition3}
\end{align}
With~\eqref{eq-condition1}, \eqref{eq-condition2} and \eqref{eq-condition3}, each edge $uw$ passes the check in chain $X$ if and only if it passes the check in chain $Y$. 
Combining with the fact that $X_u=Y_u$ for all $u\neq v_0$, this implies $X_u'=Y_u'$, a contradiction.

We then verify~\eqref{eq-condition1}, \eqref{eq-condition2} and \eqref{eq-condition3}:
\begin{itemize}
 \item
If $X_u=Y_u \in \{X_{v_0},Y_{v_0}\}$, then for every neighbor $w \in \Gamma(u)$, either $w$ is blocked or $w=v_0$. In both cases $c_w^X=c_w^Y$ is sampled consistently, this implies~\eqref{eq-condition1} and~\eqref{eq-condition2}, because $c^X_u = c^Y_u$ and $X_u = Y_u$. And  it holds that either $\{X_w,Y_w\} = \{X_{v_0},Y_{v_0}\}$ (in case of $w = v_0$) or $X_w = Y_w$ (in case of $w \neq v_0$), this implies~\eqref{eq-condition3} because $c^X_u = c^Y_u \not \in \{X_{v_0},Y_{v_0}\}$. 
  \item
  If $X_u=Y_u \not \in \{X_{v_0},Y_{v_0}\}$. For each neighbor $w \in \Gamma(u)$, it holds that either $\{c^X_w,c^Y_w\} = \{X_{v_0},Y_{v_0}\}$ or  $c^X_w=c^Y_w$, because the event $c^X_w \neq c^Y_w$ happens if and only if $\{c^X_w,c^Y_w\} = \{X_{v_0},Y_{v_0}\}$ due to the coupling. Recall that $c^X_u = c^Y_u \not \in \{X_{v_0},Y_{v_0}\}$ and $X_u=Y_u \not \in \{X_{v_0},Y_{v_0}\}$, this implies~\eqref{eq-condition1} and~\eqref{eq-condition2}. And  it holds that either $\{X_w,Y_w\} = \{X_{v_0},Y_{v_0}\}$ (in case of $w = v_0$) or $X_w = Y_w$ (in case of $w \neq v_0$), this implies~\eqref{eq-condition3} because $c^X_u = c^Y_u \not \in \{X_{v_0},Y_{v_0}\}$. 
 
\end{itemize}

For an unblocked vertex $u\neq v_0$, assume $X_u'\neq Y_u'$. By above argument, we must have $c_u^X, c_u^Y\in\{X_{v_0},Y_{v_0}\}$. We then show that $c^X_u \neq c^Y_u$. By contradiction, we assume $c^X_u = c^Y_u$, since $c_u^X, c_u^Y\in\{X_{v_0},Y_{v_0}\}$, the $(c_u^X, c_u^Y)$ must be sampled from the consistent distribution. And since $u$ is unblocked and $u\neq v_0$, the $(c_u^X, c_u^Y)$ is sampled from the consistent distribution only when for all neighbors $w\in \Gamma(u)$, $w\neq v_0$ (which means $X_w=Y_w$) and $c^X_w=c^Y_w$. In summary, $X_u=Y_u$, $c_u^X=c_u^Y$, and $X_w=Y_w$, $c_w^X=c_w^Y$ for all neighbors $w\in\Gamma(u)$, which guarantees that $X_u'=Y_u'$, a contradiction.
Therefore, we also show that for any unblocked $u\neq v_0$,  $X_u'\neq Y_u'$ only if $c_u^X\neq c_u^Y$.
 \end{proof}

We then analyze the probability of $X'_u \neq Y'_u$ for each vertex $u \in V$. 

\begin{lemma}
\label{lemma-v0}
For the vertex $v_0$ at which the $q$-colorings $X,Y \in [q]^V$ disagree, 
\begin{align*}
\Pr[X'_{v_0}=Y'_{v_0} \mid X,Y] 
\geq \left(1-\frac{\Delta}{q}\right)\QQ{\Delta}\left(1-\frac{1}{q-2}\right)^{b_{v_0}}.
\end{align*}
\end{lemma}
\begin{proof}
The event $X'_{v_0}=Y'_{v_0}$ occurs if $v_0$ accepts the proposal, which happens if the following events occur simultaneously:
\begin{itemize}
\item
$c^X_{v_0} \not \in \{X_u \mid u \in \Gamma(v_0)\}$ (and hence $c^Y_{v_0} \not \in \{Y_u \mid u \in \Gamma(v_0)\}$ by the coupling $c^Y_{v_0}=c^X_{v_0}$ and the fact that $X_u=Y_u$ for $u\neq v_0$). This occurs with probability at least $\frac{q-d_{v_0}}{q}$.
\item
For all unblocked neighbors $u \in \Gamma^U(v_0)$, it must have
$c^X_{u} \not \in \{X_{v_0},c^X_{v_0}\}$ and $c^Y_u \not \in \{Y_{v_0},c^Y_{v_0}\}$. This occurs with probability at least $\QQ{d_{v_0}-b_{v_0}}$ conditioning on any choice of $c^X_{v_0}=c^Y_{v_0}$.
\item
For all blocked neighbors $w \in \Gamma^B(v_0)$, it must have $c^X_w \not \in \{c^X_{v_0},X_{v_0},Y_{v_0}\}$ (and hence $c^Y_w \not \in \{c^Y_{v_0},X_{v_0},Y_{v_0}\}$ due to the coupling $c^Y_{w}=c^X_{w}$). This occurs with probability at least $\left(1-\frac{3}{q}\right)^{b_{v_0}}$ conditioning on any choice of $c^X_{v_0}=c^Y_{v_0}$ and independent of unblocked neighbors $u \in \Gamma^U(v_0)$.
\end{itemize}
Thus the following is obtained by the chain rule:
\begin{align*}
\Pr[X'_{v_0}=Y'_{v_0} \mid X,Y] \geq \frac{q-d_{v_0}}{q}\QQ{d_{v_0}-b_{v_0}}\left(1-\frac{3}{q}\right)^{b_{v_0}}
\ge \left(1-\frac{\Delta}{q}\right)\QQ{\Delta}\left(1-\frac{1}{q-2}\right)^{b_{v_0}},
\end{align*}
where the last inequality is due to the monotonicity stated in Lemma~\ref{lemma-inequality}.
\end{proof}

\begin{lemma}
\label{lemma-unblocked}
For any unblocked vertex $u \neq v_0$, it holds that
\begin{align}
\Pr[X'_u \ne Y'_u \mid X,Y] \leq &\frac{1}{q}\QQ{d_u-1}\left[2-\left(1-\frac{1}{q-2}\right)^{b_u}\right] \times \sum_{\text{unblocked SSAW}\atop \mathcal{P}\text{ from }v_0\text{ to }u}\left(\frac{2}{q}\right)^{\ell(\mathcal{P})-1},\label{eq:lemma-unblocked}
\end{align}
where the sum enumerates all strongly self-avoiding walks (SSAW) $\mathcal{P}=(v_0,v_1,\ldots,v_{\ell})$ from $v_0$ to $v_{\ell}=u$ over unblocked vertices $v_1,v_2,\ldots, v_{\ell}=u$, and $\ell(\mathcal{P})=\ell$ denotes the length of the walk $\mathcal{P}$.
\end{lemma}

\begin{proof}
Due to Proposition~\ref{prop:global-coupling-differ}, for unblocked $u\neq v_0$, the event $X_u'\neq Y_u'$ occurs only if  $c_u^X\neq c_u^Y$ and $u$ accepts its proposal in at least one chain among $X,Y$. 
Observe that any edge $uv$ between unblocked vertices $u,v$ either passes the check in both chains $X,Y$ or does not pass the check in both chains. 
Therefore, the event $X_u'\neq Y_u'$ occurs for an unblocked $u\neq v_0$ only if the following events occurs simultaneously:
\begin{itemize}
\item $c_u^X\neq c_u^Y$, which according to Proposition~\ref{prop:SSAW}, occurs only if there is a SSAW  $\mathcal{P}=(v_0,v_1,\ldots, v_{\ell})$ from $v_0$ to $v_\ell=u$ through  unblocked vertices $v_1,\ldots, v_{\ell}$ such that $\mathcal{P}$ is a path of disagreement;
\item for all unblocked neighbors $w\in\Gamma^{U}(u)$, the edge $uw$ passes the check, which means $c_w^X\not\in\{c_u^X,X_u\}$ (and meanwhile $c_w^Y\not\in\{c_u^Y,Y_u\}$ by coupling) for all $w\in\Gamma^{U}(u)$;
\item all blocked neighbors $w\in\Gamma^{B}(u)$ passes the check in at least one chains among $X,Y$, which means either $c^X_w \not\in \{c^X_u, X_u\}$ for all $w \in \Gamma^B(u)$ or $c^Y_w \not\in \{c^Y_u,Y_u\}$ for all $w \in \Gamma^B(u)$.
\end{itemize}
More specifically, these events occur only if: 
\begin{itemize}
\item
there is a SSAW  $\mathcal{P}=(v_0,v_1,\ldots, v_{\ell})$ from $v_0$ to $v_\ell=u$ through  unblocked vertices $v_1,\ldots, v_{\ell}$ such that $c_{v_i}^X\in\{X_{v_0},Y_{v_0}\}$ for $1\le i\le \ell-1$, which occurs with probability $\left(\frac{2}{q}\right)^{\ell-1}$; 
\item 
if $u\in \Gamma(v_0)$, then $c_u^X=Y_{v_0}$ (and meanwhile $c_u^Y=X_{v_0}\}$ by coupling), and
if $u\not\in \Gamma(v_0)$, $c_u^X\in\{X_{v_0},Y_{v_0}\}\setminus\{c_{v_{\ell-1}}^X\}$ (and meanwhile $c_u^Y\in\{X_{v_0},Y_{v_0}\}\setminus\{c_{v_{\ell-1}}^Y\}$ by coupling),
which in either case, occurs with probability $\frac{1}{q}$ conditioning on $(c_{v_{\ell-1}}^X,c_{v_{\ell-1}}^Y)$; 
\item 
$c_w^X\not\in\{c_u^X,X_u\}$ (and meanwhile $c_w^Y\not\in\{c_u^Y,Y_u\}$ by coupling) for all unblocked $w\in\Gamma^U(u)\setminus\{v_{\ell-1}\}$, which occurs with probability $\left(1-\frac{2}{q}\right)^{d_u-b_u-1}$ conditioning on $c_u^X$;
\item 
either $c^X_w \not\in \{c^X_u, X_u\}$ for all $w \in \Gamma^B(u)$ or $c^Y_w \not\in \{c^Y_u,Y_u\}$ for all $w \in \Gamma^B(u)$, which occurs with probability at most $\left[2\QQ{b_u}-\left(1-\frac{3}{q}\right)^{b_u}\right]$ conditioning on $(c_u^X,c_u^Y)$ by the principle of inclusion-exclusion.
\end{itemize}
Take the union bound over all SSAW $\mathcal{P}=(v_0,v_1,\ldots, v_{\ell})$ through  unblocked vertices $v_1,\ldots, v_{\ell}=u$. Due to the strongly-avoiding property, it is safe to apply the chain rule for every $\mathcal{P}$. We have:
\begin{align*}
&\Pr[X'_u \ne Y'_u \mid X,Y]\\
\leq\,& \sum_{\text{unblocked SSAW}\atop \mathcal{P}\text{ from }v_0\text{ to }u}\Bigg(\left(\frac{2}{q}\right)^{\ell(\mathcal{P})-1}\left(\frac{1}{q}\right)\QQ{d_u-b_u-1}\times\left[2\QQ{b_u}-\left(1-\frac{3}{q}\right)^{b_u}\right]\Bigg)\\
=\,&\frac{1}{q}\QQ{d_u-1}\left[2-\left(1-\frac{1}{q-2}\right)^{b_u}\right]\sum_{\text{unblocked SSAW}\atop \mathcal{P}\text{ from }v_0\text{ to }u}\left(\frac{2}{q}\right)^{\ell(\mathcal{P})-1}. \qedhere
\end{align*}
\end{proof}

\begin{lemma}
\label{lemma-blocked}
For any blocked vertex $u \neq v_0$, it holds that
\begin{align}
\label{eq:lemma-blocked}
\Pr[X'_u \neq Y'_u \mid X, Y] \leq &\frac{1}{q}\QQ{d_u-1}\sum_{\text{SSAW }\mathcal{P}\text{ from } v_0 \text{ to } u \atop  \text{with only } u \text{ blocked}}\left(\frac{2}{q} \right)^{\ell(\mathcal{P})-1},
\end{align}
where the sum enumerates all the strongly self-avoiding walks (SSAW) $\mathcal{P}=(v_0,v_1,\ldots,v_{\ell})$ from $v_0$ to $v_{\ell}=u$ through unblocked vertices   $v_1,\ldots,v_{\ell-1}$, and $\ell(\mathcal{P})=\ell$ denotes the length of the walk $\mathcal{P}$. 
\end{lemma}

\begin{proof}
By the coupling, any blocked vertex $u \in V$ proposes consistently in the two chains, thus $c^X_u = c^Y_u$. And we have $X_u = Y_u$ for $u \ne v_0$.

We first consider $v_0$'s blocked neighbors $u \in \Gamma^B(v_0)$. There are two cases for such vertex $u$:
\begin{itemize}
\item
$X_u = Y_u \not \in\{X_{v_0},Y_{v_0}\}$. Since vertex $u$ is blocked, there must exist a vertex $w_0 \in \Gamma(u) \setminus \{v_0\}$, such that $X_{w_0} = Y_{w_0} \in \{X_{v_0},Y_{v_0}\}$. Without loss of generality, suppose $X_{w_0} = Y_{w_0} =X_{v_0}$ (and the case $X_{w_0}=Y_{w_0}=Y_{v_0}$ follows by symmetry). By Proposition~\ref{prop:global-coupling-differ}, $X'_u \ne Y'_u$ only if $c^X_u = c^Y_u \in \{X_{v_0},Y_{v_0}\}$. Note that if $c^X_u = c^Y_u = X_{v_0}$, then the edge $uw_0$ cannot pass the check in both chains, hence $X'_u = Y'_u$, a contradiction. So we must have $c^X_u = c^Y_u=Y_{v_0}$, in which case edge $v_0u$ cannot pass the check in chain $Y$, thus the event $X'_u \ne Y'_u$ occurs only when $u$ accepts the proposal in chain $X$, which happens only if for all $w \in \Gamma(u)$, $c^X_w \not \in \{c^X_u,X_u\}$. Remember that we already have $c^X_u=Y_{v_0}\neq X_u$ and note that all vertices in chain $X$ propose independently, therefore $X'_u \neq Y'_u$ occurs with probability at most $\frac{1}{q}\QQ{d_u}$.
\item
$X_u = Y_u \in \{X_{v_0},Y_{v_0}\}$. Without loss of generality, suppose $X_u = Y_u = X_{v_0}$(and the case $X_u = Y_u = Y_{v_0}$ follows by symmetry). By Proposition~\ref{prop:global-coupling-differ}, $X'_u \ne Y'_u$ only if $c^X_u = c^Y_u \in \{X_{v_0},Y_{v_0}\}$. If $c^X_u=c^Y_u=X_{v_0}$, the proposal and the current color of $u$ are the same in two chains, hence $X'_u = Y'_u$, a contradiction. So we must have $c^X_u = c^Y_u=Y_{v_0}$, in which case the edge $uv_0$ cannot pass the check in chain $Y$, thus  event $X'_u \ne Y'_u$ occurs only if vertex $u$ accepts the proposal in chain $X$, which happens only if for all $w \in \Gamma(u)$, $c^X_w \not \in \{c^X_u,X_u\}=\{X_{v_0},Y_{v_0}\}$. 
Remember that we already have $c^X_u=Y_{v_0}$ and note that all vertices in chain $X$ propose independently, therefore $X'_u \neq Y'_u$ occurs with probability at most $\frac{1}{q}\QQ{d_u}$. 
\end{itemize}
Hence, for all $u \in \Gamma^B(v_0)$, we have:
\begin{align*}
\Pr[X'_u \neq Y'_u \mid X,Y] \leq \frac{1}{q}\QQ{d_u} \leq \frac{1}{q}\QQ{d_u-1}. 
\end{align*}
The walk $\mathcal{P}=(v_0,u)$ is a strongly self-avoiding walk (SSAW) from $v_0$ to $u$ with only $u$ blocked. Therefore~\eqref{eq:lemma-blocked} is proved for blocked vertices $u \in \Gamma^B(v_0)$.

Now we consider the general blocked vertices $u \not \in \Gamma^+(v_0)$. Assume that $X_u'\neq Y'_u$.

If $u$ is blocked by itself, i.e.~$X_u = Y_u \in \{X_{v_0},Y_{v_0}\}$, then all the vertices $w \in \Gamma^+(u)$ are blocked and hence propose consistently, and for $u\not\in\Gamma^+(v_0)$ all neighbors $w$ have $X_w=Y_w$, so we must have $X'_u = Y'_u$. Thus $\Pr[X'_u \ne Y'_u \mid X, Y] = 0$ and~\eqref{eq:lemma-blocked} holds trivially.

If otherwise $u$ is not blocked by itself, i.e.~$X_u = Y_u \not \in \{X_{v_0},Y_{v_0}\}$, then $u$ must be blocked by one of its neighbors $w_0 \in \Gamma(u)$ such that $X_{w_0}=Y_{w_0} \in \{X_{v_0},Y_{v_0}\}$. By Proposition~\ref{prop:global-coupling-differ},  $X'_u \ne Y'_u$ only if $c^X_u = c^Y_u \in\{X_{v_0},Y_{v_0}\}$. We must have $c_u^X\neq X_{w_0}$, because if otherwise $c_u^X= X_{w_0}$, together with that $c^Y_u = Y_{w_0}$ which is due to that $c^X_u=c^Y_u$ and $X_{w_0}=Y_{w_0}$,   the edge $uw_0$ cannot pass the check in both chains, giving us $X'_u = Y'_u$, a contradiction.

For the following, we assume $c^X_u = c^Y_u  \in \{X_{v_0},Y_{v_0}\}$ and $c^X_u \ne X_{w_0}$, therefore $c^Y_u \ne Y_{w_0}$ because $c^X_u=c^Y_u$ and $X_{w_0}=Y_{w_0}$. We claim that $u$ must have an unblocked neighbor $w^* \in \Gamma(u)$ such that $c_{w^*}^X\neq c_{w^*}^Y$ because if otherwise for all the vertices $w \in \Gamma^+(u)$, the consistencies $c^X_w=c^Y_w$ and $X_w = Y_w$ hold, giving us $X'_u = Y'_u$, a contradiction. Therefore,  there is a neighbor $w^* \in \Gamma(u)$ such that $c_{w^*}^X\neq c_{w^*}^Y$, which by Proposition~\ref{prop:SSAW}, means that there is a strongly self-avoiding walk (SSAW) $\mathcal{P}=(v_0,v_1,\ldots,v_{\ell})$ from $v_0$ to $v_{\ell}=u$ through unblocked $v_1,v_2,\ldots,v_{\ell-1}=w^*$ such that $\mathcal{P}'=(v_0,v_1,\ldots,v_{\ell-1})$ is a path of disagreement. Fix any SSAW $\mathcal{P}=(v_0,v_1,\ldots,v_{\ell})$ from $v_0$ to $v_{\ell}=u$ with only $u$ blocked. By Proposition~\ref{prop:SSAW}:
\begin{itemize}
\item $\mathcal{P}'=(v_0,v_1,\ldots,v_{\ell-1})$ is a path of disagreement with probability at most $\left(\frac{2}{q}\right)^{\ell-1}$.
\end{itemize}
As argued above, assuming $X_u'\neq Y'_u$ we must have
\begin{itemize}
\item $c^X_u\in\{X_{v_0},Y_{v_0}\} \setminus \{X_{w_0}\}$ (and $c^Y_u=c^X_u$ due to the coupling), which occurs with probability $\frac{1}{q}$ conditioning on that $\mathcal{P}'$ is a path of disagreement.
\end{itemize}
As argued above, we have $\{c^X_{v_{\ell-1}},c^Y_{v_{\ell-1}}\}=\{c^X_u,X_{w_0}\}=\{c^Y_u,Y_{w_0}\}=\{X_{v_0},Y_{v_0}\}$. Without loss of generality, suppose $c^X_{v_{\ell-1}} = c^X_u=c^Y_u$ and $c^Y_{v_{\ell-1}} = X_{w_0}=Y_{w_0}$ (and the case $c^X_{v_{\ell-1}} = X_{w_0} = Y_{w_0}$ and $c^Y_{v_{\ell-1}} = c^X_u = c^Y_u$ follows by symmetry). Then edge $uv_{\ell-1}$ cannot pass the check in chain $X$ because $c^X_{v_{\ell-1}}=c^X_{u}$. Then the event $X'_u \ne Y'_u$ occurs only if vertex $u$ accepts the proposal in chain $Y$, which happens only if 
\begin{itemize}
\item 
$c^Y_w \not \in \{Y_u, c^Y_u\}$ for all $w \in \Gamma(u)\setminus \{v_{\ell-1}\}$.
Recall that $Y_u \neq c^Y_u$. Since $\mathcal{P}$ is a strongly self-avoiding, we have $w \not \in \mathcal{P}$ for all $w \in \Gamma(u)\setminus \{v_{\ell-1}\}$. And the proposals are mutually independent in one chain. Condition on previous events, this probability is at most $\left(1-\frac{2}{q}\right)^{d_u - 1}$.
\end{itemize}
By the union bound over all SSAW $\mathcal{P}$ from $v_0$ to $u$ with $u$ being the only blocked vertex, and the chain rule for every $\mathcal{P}$, we have
\begin{align*}
\Pr[X'_u \neq Y'_u \mid X, Y] \leq &\frac{1}{q}\QQ{d_u-1}\sum_{\text{SSAW }\mathcal{P}\text{ from } v_0 \text{ to } u \atop  \text{with only } u \text{ blocked}}\left(\frac{2}{q} \right)^{\ell(\mathcal{P})-1}.
\end{align*}
This proves \eqref{eq:lemma-blocked}.
\end{proof}

We then verify the path coupling condition: for some constant $\delta>0$, 
\begin{align}
\label{eq-target}
\mathbf{E}[\Phi(X',Y') \mid X,Y] \leq (1-\delta)\Phi(X,Y).
\end{align}
By the linearity of expectation,
\begin{align*}
&\mathbf{E}[\Phi(X',Y') \mid X, Y]
= \sum_{u \in V}\mathbf{E}[\phi_u(X', Y') \mid X, Y]\\
=\,& d_{v_0}\Pr[X'_{v_0} \ne Y'_{v_0} \mid X,Y] +\sum_{\text{unblocked }\atop u \ne v_0}d_u\Pr[X'_{u} \ne Y'_{u} \mid X,Y] +\sum_{\text{blocked }\atop w\neq v_0}d_w\Pr[X'_{w} \ne Y'_{w} \mid X,Y]
\end{align*}
Due to Lemma~\ref{lemma-v0}, 
\begin{align}
\mathbf{E}[\phi_{v_0}(X', Y') \mid X, Y]=&d_{v_0}\Pr[X_{v_0}'\neq Y'_{v_0}\mid X,Y]\notag\\
\le& d_{v_0}\left[1- \left(1-\frac{\Delta}{q}\right)\QQ{\Delta}\left(1-\frac{1}{q-2}\right)^{b_{v_0}}\right].\label{eq:v0-phi-upper-bound}
\end{align}
On the other hand, due to Lemma~\ref{lemma-unblocked} and Lemma~\ref{lemma-blocked},
\begin{align}
&\sum_{u\neq v_0}\mathbf{E}[\phi_u(X', Y') \mid X, Y]\nonumber\\
\le&
\sum_{\text{unblocked }\atop u \ne v_0}\Bigg(
\frac{d_u}{q}\QQ{d_u-1}\left[2-\left(1-\frac{1}{q-2}\right)^{b_u}\right]\times\sum_{\text{unblocked SSAW}\atop \mathcal{P}\text{ from }v_0\text{ to }u}\left(\frac{2}{q}\right)^{\ell(\mathcal{P})-1}\notag\Bigg)\\
&
+\sum_{\text{blocked }\atop u\neq v_0}
\frac{d_u}{q}\QQ{d_u-1}
\sum_{\text{SSAW } \mathcal{P}\text{ from }v_0\text{ to }u\atop\text{with only $u$ blocked}}\left(\frac{2}{q}\right)^{\ell(\mathcal{P})-1}\notag\\
\le& 
\sum_{\text{unblocked }\atop u \ne v_0}\Bigg(
\frac{\Delta}{q}\QQ{\Delta-1}\left[2-\left(1-\frac{1}{q-2}\right)^{b_u}\right]\times\sum_{\text{unblocked SSAW}\atop \mathcal{P}\text{ from }v_0\text{ to }u}\left(\frac{2}{q}\right)^{\ell(\mathcal{P})-1}\Bigg)\notag\\
&+\sum_{\text{blocked }\atop u\neq v_0}
\frac{\Delta}{q}\QQ{\Delta-1}
\sum_{\text{SSAW } \mathcal{P}\text{ from }v_0\text{ to }u\atop\text{with only $u$ blocked}}\left(\frac{2}{q}\right)^{\ell(\mathcal{P})-1}\label{eq:phi-sum-upper-bound-monotone}\\
\le& 
\sum_{\mathcal{P}\text{ from }v_0\atop\text{ to any }u\neq v_0}
\phi_{\mathcal{P}},\label{eq:phi-sum-upper-bound}
\end{align}
where the inequality~\eqref{eq:phi-sum-upper-bound-monotone} is due to the monotonicity stated in Lemma~\ref{lemma-inequality}, and the last sum in~\eqref{eq:phi-sum-upper-bound} enumerates all the walks $\mathcal{P}=(v_0,v_1,\ldots,v_\ell)$ from $v_0$. And for such walk $\mathcal{P}$, the quantity $\phi_{\mathcal{P}}$ is defined as that $\phi_{\mathcal{P}}=0$ if $\mathcal{P}$ is not a strongly self-avoiding walk (SSAW), and for a SSAW $\mathcal{P}=(v_0,v_1,\ldots,v_\ell)$ from $v_0$ to any $v_\ell=u$:
\begin{align*}
&\phi_{\mathcal{P}}
=
\begin{cases}
\frac{\Delta}{q}\QQ{\Delta-1}\left[2-\left(1-\frac{1}{q-2}\right)^{b_u}\right]\left(\frac{2}{q}\right)^{\ell-1} & \text{(\uppercase\expandafter{\romannumeral1 })} \\
\frac{\Delta}{q}\QQ{\Delta-1}\left(\frac{2}{q}\right)^{\ell-1} &\text{(\uppercase\expandafter{\romannumeral2 })} \\
0 & \text{(\uppercase\expandafter{\romannumeral3 })} 
\end{cases}\\
&\text{\uppercase\expandafter{\romannumeral1 }}: \text{if all }v_1,\ldots, v_\ell\text{ are unblocked};\\
&\text{\uppercase\expandafter{\romannumeral2 }}: \text{if all }v_1,\ldots, v_{\ell-1}\text{ are unblocked and $v_\ell=u$ is blocked};\\
&\text{\uppercase\expandafter{\romannumeral3 }}: \text{otherwise}.
\end{align*}
It is easy to verify the inequality~\eqref{eq:phi-sum-upper-bound} with this definition of $\phi_{\mathcal{P}}$.

Given any walk $\mathcal{P}=(v_0,v_1,\ldots,v_\ell)$ from $v_0$ such that all $v_1,\ldots,v_\ell$ are unblocked, we further define that
\begin{align}
\Phi_{\mathcal{P}}
=\left(\frac{q}{2}\right)^{\ell-1}\sum_{\mathcal{P}'\text{ extends }\mathcal{P}}\phi_{\mathcal{P}'},
\label{eq:definition-Phi-P}
\end{align}
where the sum enumerates all walks (not necessarily strongly self-avoiding) $\mathcal{P}'=(v_0,v_1,\ldots,v_\ell, v_{\ell+1}, \ldots)$ with  $\mathcal{P}$ as its prefix, including $\mathcal{P}$ itself. 

Then by the inequality~\eqref{eq:phi-sum-upper-bound} the expected distance except for $v_0$ can be expressed as:
\begin{align}
\sum_{u\neq v_0}\mathbf{E}[\phi_u(X', Y') \mid X, Y]\le& 
\sum_{\mathcal{P}\text{ from }v_0\atop\text{ to any }u\neq v_0}\phi_{\mathcal{P}}=
\sum_{u\in\Gamma(v_0)\setminus\Gamma^B(v_0)}\Phi_{(v_0,u)}
+\sum_{u\in\Gamma^B(v_0)}\phi_{(v_0,u)}\notag\\
=&
\sum_{u\in\Gamma(v_0)\setminus\Gamma^B(v_0)}\Phi_{(v_0,u)}
+\frac{\Delta b_{v_0}}{q}\QQ{\Delta-1}.\label{eq:Phi-sum-upper-bound}
\end{align}
Here each $(v_0,u)$ is a path (of length 1) from $v_0$ to its neighbor $u$.

And more importantly, for $\Phi_{\mathcal{P}}$ we have the following recurrence.
For any walk ${\mathcal{P}}=(v_0,v_1,\ldots,v_\ell)$ from $v_0$ through unblocked vertices $v_1,\ldots,v_\ell=u$, if $\mathcal{P}$ is not strongly self-avoiding then $\Phi_{\mathcal{P}}=0$; and if otherwise ${\mathcal{P}}$ is strongly self-avoiding, then the following recurrence  follows directly from the definition~\eqref{eq:definition-Phi-P} of $\Phi_{\mathcal{P}}$:
\begin{align}
\Phi_{\mathcal{P}}
&=
\left(\frac{q}{2}\right)^{\ell-1}\phi_{\mathcal{P}}+\left(\frac{q}{2}\right)^{\ell-1}\sum_{w\in\Gamma^B(u)}\phi_{(\mathcal{P},w)}+\frac{2}{q}\sum_{\text{unblocked }w\in\Gamma(u)\atop w\neq v_{\ell-1}}\Phi_{(\mathcal{P},w)}\notag\\
&\le \frac{\Delta}{q}\QQ{\Delta-1}\left[2-\left(1-\frac{1}{q-2}\right)^{b_u}+\frac{2b_u}{q}\right]+\frac{2}{q}\sum_{\text{unblocked }w\in\Gamma(u)\atop w\neq v_{\ell-1}}\Phi_{(\mathcal{P},w)},\label{eq:Phi-recursion}
\end{align}
where $(\mathcal{P},w)$ denotes the walk $\mathcal{P}'=(v_0,v_1,\ldots,v_{\ell},w)$ that extends $\mathcal{P}$.

The following lemma essentially states that $\Phi_{\mathcal{P}}$ is maximized when the number of blocked neighbors $b_u=0$ and then the value of $\Phi_{\mathcal{P}}$ is upper bounded by the fixpoint for this recurrence.
 
\begin{lemma}
\label{lemma-PHI}
If $3\Delta < q \leq 3.7\Delta+3$ and $\Delta \geq 5$, then for any walk  $\mathcal{P}=(v_0,v_1,\ldots,v_\ell)$ from $v_0$ such that all $v_1,\ldots,v_\ell$ are unblocked, it holds that
\begin{align*}
\Phi_{\mathcal{P}} \leq \frac{\Delta}{q-2\Delta+2}\QQ{\Delta-1}.  
\end{align*}
\end{lemma}
\begin{proof}
We prove by induction on the length of the walk. Let $\mathcal{P}=(v_0,v_1,\ldots,v_\ell)$ be a walk from $v_0$ such that all $v_1,\ldots,v_\ell$ are unblocked and $v_\ell=u$. 
When $\ell$ is longer than the longest strongly self-avoiding walk among unblocked $v_1,\ldots,v_\ell$, then $\mathcal{P}$ is not a SSAW and thus $\Phi_{\mathcal{P}}=0$.

Assume that the lemma holds for all unblocked walks longer than $\ell$. Then due to the recurrence~\eqref{eq:Phi-recursion}, 
\begin{align*}
\Phi_{\mathcal{P}}
&\le \frac{\Delta}{q}\QQ{\Delta-1}\left[2-\left(1-\frac{1}{q-2}\right)^{b_u}+\frac{2b_u}{q}\right]+\frac{2}{q}\sum_{\text{unblocked }w\in\Gamma(u)\atop w\neq v_{\ell-1}}\Phi_{(\mathcal{P},w)}\\
\text{(I.H.)}\quad &\le \frac{\Delta}{q}\QQ{\Delta-1}\left[2-\left(1-\frac{1}{q-2}\right)^{b_u}+\frac{2b_u}{q}\right]+\frac{2(\Delta-b_u-1)\Delta}{q(q-2\Delta+2)}\QQ{\Delta-1}\\
&=\bigg[1-\left(1-\frac{1}{q-2}\right)^{b_u}-\frac{4\Delta-4}{q(q-2\Delta+2)}\cdot b_u+\frac{\Delta}{q-2\Delta+2}\cdot\frac{q}{\Delta}\bigg]\frac{\Delta}{q}\QQ{\Delta-1},
\end{align*}
which is bounded from above by $\frac{\Delta}{q-2\Delta+2}\QQ{\Delta-1}$ if
\[
\left(1-\frac{1}{q-2}\right)^{b_u}+\frac{4\Delta-4}{q(q-2\Delta+2)}\cdot b_u\ge 1.
\]
The inequality holds trivially when $b_u = 0$. It is then sufficient to prove that LHS is monotone on integer $b_u\ge 0$: Denoted $f(x)=\left(1-\frac{1}{q-2}\right)^{x}+\frac{4\Delta-4}{q(q-2\Delta+2)}\cdot x$, 
\begin{align*}
f(b_u+1)-f(b_u)
&=\frac{4\Delta-4}{q(q-2\Delta+2)}-\left(1-\frac{1}{q-2}\right)^{b_u}\left(\frac{1}{q-2}\right)\\
(\text{since }b_u \geq 0)\qquad
&\geq \frac{4\Delta-4}{q(q-2\Delta+2)}-\frac{1}{q-2},
\end{align*}
which is nonnegative for $3\Delta-3-\sqrt{9\Delta^2-26\Delta+17}\le q\le 3\Delta-3+\sqrt{9\Delta^2-26\Delta+17}$. In particular this holds when $3\Delta < q \leq 3.7\Delta+3$ and $\Delta \geq 5$. This completes the induction.
\end{proof}

\paragraph{Proof of Lemma~\ref{lemma-coloring-global-coupling}:}
Combine~\eqref{eq:v0-phi-upper-bound} and~\eqref{eq:Phi-sum-upper-bound}, with Lemma~\ref{lemma-PHI}, we obtain
\begin{align}
&\mathbf{E}[\Phi(X',Y') \mid X, Y]=\sum_{u \in V}\mathbf{E}[\phi_u(X', Y') \mid X, Y]\notag\\
\le\,& d_{v_0}\left[1- \left(1-\frac{\Delta}{q}\right)\QQ{\Delta}\left(1-\frac{1}{q-2}\right)^{b_{v_0}}\right]\notag\\
&
+\frac{\Delta(d_{v_0}-b_{v_0})}{q-2\Delta+2}\QQ{\Delta-1}
+\frac{\Delta b_{v_0}}{q}\QQ{\Delta-1}.\label{eq:global-path-coupling-upper-bound}
\end{align}
We need the following technical inequality:
\begin{align}
\left(1-\frac{\Delta}{q}\right)&\le \left(1-\frac{\Delta}{q}\right)\left(1-\frac{1}{q-2}\right)^{b_{v_0}}+ \frac{2\Delta-2}{(q-2)(q-2\Delta+2)} b_{v_0}\label{eq:global-coupling-b-v0-monotone}
\end{align}
The equality holds trivially when $b_{v_0}=0$. It is then sufficient to verify that the RHS is monotone on integer $b_{v_0}\ge 0$. We denote $g(x)=\left(1-\frac{\Delta}{q}\right)\left(1-\frac{1}{q-2}\right)^{x} + \frac{2\Delta-2}{(q-2)(q-2\Delta+2)}x$, and
\begin{align*}
g(b_{v_0}+1)-g(b_{v_0})=&\frac{q}{(q-2\Delta+2)(q-2)}-\frac{1}{q-2}-\left(1-\frac{\Delta}{q}\right)\left(1-\frac{1}{q-2}\right)^{b_{v_0}}\frac{1}{q-2}\\
\geq& \frac{q}{(q-2\Delta+2)(q-2)}-\frac{1}{q-2}-\frac{q-\Delta}{q(q-2)},
\end{align*}
which is nonnegative if $\frac{q}{(q-2\Delta+2)}\geq 1 + \frac{q-\Delta}{q}$. This easily holds for $\frac{1}{2}(5\Delta-4-\sqrt{17\Delta^2-32\Delta+16})\le q\le\frac{1}{2}(5\Delta-4+\sqrt{17\Delta^2-32\Delta+16})$. In particular, it holds as long as $\Delta \le q \le 3.7\Delta+3$ and $\Delta\ge9$.

With the inequality~\eqref{eq:global-coupling-b-v0-monotone}, the RHS in~\eqref{eq:global-path-coupling-upper-bound} is maximized when $b_0=0$ and hence
\begin{align*}
\mathbf{E}[\Phi(X',Y') \mid X, Y]\le
d_{v_0}\left[1-\left(1-\frac{\Delta}{q}\right)\QQ{\Delta} + \frac{\Delta}{q-2\Delta+2}\QQ{\Delta-1}\right].
\end{align*}
Recall that $\Phi(X,Y)=d_{v_0}$. The path coupling condition~\eqref{eq-target} holds when there is a constant $\delta>0$ such that 
\begin{align}
\left(1-\frac{\Delta}{q}\right)\QQ{\Delta} - \frac{\Delta}{q-2\Delta+2}\QQ{\Delta-1}
\ge \delta.\label{eq:global-path-coupling-constant-delta}
\end{align}
For $q=\alpha^{\star}\Delta$ and $\Delta\to\infty$, then the LHS becomes $\mathrm{e}^{-2/\alpha^{\star}}\left(1-\frac{1}{\alpha^{\star}}-\frac{1}{\alpha^{\star}-2}\right)$, which equals $0$ if $\alpha^{\star}=2+\sqrt{2}$.

Furthermore, for $q\ge\alpha\Delta$, the LHS become:
\begin{align*}
\left(1-\frac{\Delta}{q}\right)\QQ{\Delta} - \frac{\Delta}{q-2\Delta+2}\QQ{\Delta-1}
&\ge\left(1-\frac{2}{q}\right)^\Delta\left(1-\frac{\Delta}{q}-\frac{\Delta}{q-2\Delta}\right)
\\&\ge\left(1-\frac{2}{\alpha\Delta}\right)^\Delta\left(1-\frac{1}{\alpha}-\frac{1}{\alpha-2}\right)
\\&\ge\left(1-\frac{2}{\alpha}\right)\left(1-\frac{1}{\alpha}-\frac{1}{\alpha-2}\right)
\end{align*}
which is a positive constant independent of $\Delta$ when $\alpha>\alpha^{\star}=2+\sqrt{2}$.

Altogether, by the path coupling Lemma~\ref{pathcoupling}, if $\alpha\Delta\le q\le3.7\Delta+3$ for a constant $\alpha > 2 + \sqrt{2}$ and $\Delta \geq 9$, then the mixing rate is bounded by $\tau(\epsilon) = O(\log\left(\frac{n}{\epsilon}\right))$.

\section{Lower bounds}\label{sec:lower-bound}
In this section, we show lower bounds for local sampling. 
Let $G(V,E)$ be a network, and $\mathcal{I}$ an instance of MRF or weighted local CSP defined on graph $G$. For example, $\mathcal{I}=(G,[q],\boldsymbol{A},\boldsymbol{b})$ for a MRF with edge activities $\boldsymbol{A}=\{A_e\}_{e\in E}$ and vertex activities $\boldsymbol{b}=\{b_v\}_{v\in V}$.

We assume that each vertex $v\in V$ may access to an independent random variable $\Psi_v$ as its source of randomness.
Then a $t$-round protocol specifies a family of functions $\Pi_{v,\mathcal{I}}$, such that for each vertex $v\in V$, the output $X_v$ is produced as
\[
X_v=\Pi_{v,\mathcal{I}}(\Psi_u, u\in B_t(v)),
\]
where $B_t(v)=\{u\in V\mid\dist(u,v)\le t\}$ represents the $t$-ball centered at $v$.
Let $\mu_{\mathsf{out}}$ denote the distribution of the output random vector $\boldsymbol{X}=(X_v)_{v\in V}$. The goal is to have $\DTV{\mu_{\mathsf{out}}}{\mu}\le \epsilon$, where $\mu=\mu_{\mathcal{I}}$ is the Gibbs distribution  defined by the MRF instance $\mathcal{I}$.

Note that in above we allow the protocol $\Pi_{v,\mathcal{I}}$ executed at each vertex $v\in V$ to be aware of the instance $\mathcal{I}$ of the MRF.
This is much stronger than the original \LOCAL{} model. In fact, the only locality property we are using to prove our lower bounds is that for any $\boldsymbol{X}=(X_v)_{v\in V}$ returned by a $t$-round protocol:
\begin{align}
\forall u,v\in V: \quad\dist(u,v)> 2t \Longrightarrow X_u\text{ and }X_v\text{ are independent}.\label{eq:non-local-independence}
\end{align}
The lower bounds implied by this property is due to the locality of randomness.

For many natural MRFs, the Gibbs distribution $\mu$ exhibits the following exponential correlations: There exist constants $\delta,\eta>0$ such that for a path $P$ of length $n$, any vertices $u,v$ from the path, there are two spin states $\sigma_u,\sigma_u'\in[q]$ such that 
$\mu_u(\sigma_u)\ge\delta,\mu_u(\sigma_u')\ge \delta$  for the marginal distribution $\mu_u$ induced by $\mu$ at vertex $u$
and
\begin{align}
\DTV{\mu_v(\cdot \mid \sigma_u)}{\mu_v(\cdot\mid\sigma_u')}\ge \eta^{\dist(u,v)}.
\label{eq:exponential-correlation}
\end{align}
This exponential correlation property is satisfied by many MRFs, in particular, the proper $q$-colorings for any constant $q$.
For MRFs having this property, for any $\epsilon>\exp(-o(n))$, vertex pairs $(u,v)$ with sufficiently small $\dist(u,v)=\Omega(\log \frac{1}{\epsilon})$ will contribute at least an $\epsilon$ total variation distance between Gibbs $(\sigma_u,\sigma_v)$ and any independent $(X_u,X_v)$. And due to~\eqref{eq:non-local-independence}, this gives an $\Omega(\log \frac{1}{\epsilon})$ lower bound for local sampling from any MRF satisfying~\eqref{eq:exponential-correlation}, where $\epsilon$ is the total variation distance. 

We then show that the $\Omega(\log n)$ lower bound holds even for a constant total variation distance~$\epsilon$. 
A similar $\Omega(\log n)$ lower bound for sampling independent sets is proved independently in~\cite{guo2016uniform}.
Altogether it shows that the $O\left(\log \left(\frac{n}{\epsilon}\right)\right)$ upper bound in Theorem~\ref{main-theorem-LocalMetropolis} is optimal.


\begin{theorem}
Let $q\ge 3$ be a constant and $\epsilon<\frac{1}{3}$. Any $t$-round protocol that samples uniform proper $q$-coloring in a path within total variation distance $\epsilon$ must have $t=\Omega(\log n)$.
\end{theorem}
\begin{proof}
We actually prove the lower bound for all MRFs satisfying a stronger exponential correlation property stated as follows: There exist constants $\delta,\eta>0$ such that for a path $P$ of length $n$, for any non-adjacent vertices $x,u,v,y$ in the path from left to right, any spin states $\sigma_x,\sigma_y\in[q]$, there exist two spin states $\sigma_u,\sigma_u'\in[q]$ such that $\mu_u(\sigma_u\mid \sigma_x,\sigma_y)\ge\delta,\mu_u(\sigma_u'\mid \sigma_x,\sigma_y)\ge \delta$ and
\begin{align}
\DTV{\mu_v(\cdot \mid \sigma_u, \sigma_x,\sigma_y)}{\mu_v(\cdot\mid\sigma_u',\sigma_x,\sigma_y)}\ge \eta^{\dist(u,v)}.
\label{eq:exponential-correlation-cond}
\end{align}
It can be verified by a simple recursion for marginal probabilities in paths~\cite{lu2013improved} that this property as well as the weaker correlation property~\eqref{eq:exponential-correlation} hold for uniform proper $q$-colorings in paths for any constant $q\ge 3$.

Let $P=(w_0,w_1,\ldots,w_{n-1})$ be a path of $n$ vertices. For $i=0,1,\ldots,m$ where $m=\left\lfloor\frac{n-1}{3(2t+1)}\right\rfloor$, we denote $x_i=w_{3(2t+1)i}$; and for $i=0,1,\ldots,m-1$, denote $u_i=w_{3(2t+1)i+2t+1}$, and $v_i=w_{3(2t+1)i+2(2t+1)}$. 
We denote $F=\{x_i\mid 0\le i\le m\}$ and $U=\{u_i,v_i\mid 0\le i\le m-1\}$, and let $C=F\cup U$. We call the vertices in $C$ the \concept{centers}, and the vertices in $F$ and $U$ the \concept{fixed} and \concept{unfixed} centers respectively.
Note that the pairs $(u_i,v_i)$ of consecutive unfixed centers are separated by the fixed centers $x_i$'s. Due to the conditional independence of MRF,  conditioning on any particular configuration $\sigma_F\in[q]^F$ of fixed centers, for a $\sigma\in[q]^P$ sampled from the Gibbs distribution $\mu$ consistent with $\sigma_F$ over $F$, the pairs $(\sigma_{u_i},\sigma_{v_i})$ are mutually independent of each other. For the followings we assume that we are conditioning on an arbitrarily fixed $\sigma_F\in[q]^F$.

Let ${X}_{u_i}$ and $X_{v_i}$ be the respective output of $u_i$ and $v_i$ in a $t$-round protocol. Due to the observation of~\eqref{eq:non-local-independence}, ${X}_{u_i}$ and $X_{v_i}$ are mutually independent. According to the exponential correlation of~\eqref{eq:exponential-correlation-cond}, by choosing a suitably small $t=O(\log n)$, the total variation distance between $(\sigma_{u_i},\sigma_{v_i})$ and $({X}_{u_i},{X}_{v_i})$ is at least $\exp(-\Omega(t)) =n^{-\frac{1}{4}}$. 

We denote $\mathcal{X}_i=({X}_{u_i},{X}_{v_i})$ and $\mathcal{Y}_i=({\sigma}_{u_i},{\sigma}_{v_i})$, and consider the random vector $\mathcal{X}=(\mathcal{X}_i)_{0\le i\le m-1}$ and $\mathcal{Y}=(\mathcal{Y}_i)_{0\le i\le m-1}$ where $\mathcal{Y}$ is sampled conditioning on an arbitrarily fixed $\sigma_F\in[q]^F$. As we argued above, both $\mathcal{X}=(\mathcal{X}_i)$ and $\mathcal{Y}=(\mathcal{Y}_i)$ are vectors of mutually independent variables, and $\DTV{\mathcal{X}_i}{\mathcal{Y}_i}\ge n^{-\frac{1}{4}}$.
Suppose $\mathcal{X}$ follows the product distribution $\pi = \pi_0 \times \pi_1\times\ldots \times \pi_{m-1}$, where $\pi_i$ is a distribution over $[q]^2$.
Suppose $\mathcal{Y}$ follows the product distribution $\nu = \nu_0 \times \nu_1 \times \ldots \times\nu_{m-1}$, where $\nu_i$ is a distribution over $[q]^2$.
Given $\pi$ and $\nu$, for each $0 \leq i \leq m -1$, define a map $f_i: [q]^2 \to \{0,1\}$:
\begin{align*}
	\forall c\in [q]^2,\quad f_i(c) \triangleq \begin{cases}
		1 &\text{if } \pi_i(c) > \nu_i(c)\\
		0 &\text{if } \pi_i(c) \leq \nu_i(c).
	\end{cases}
\end{align*}
Define the function $f: ([q]^2)^m \to \mathbb{N}$ by
\begin{align*}
	\forall \tau \in ([q]^2)^m , \quad f(\tau) \triangleq \sum_{i = 0}^{m-1}f_i(\tau_i), \text{ where } \tau_i \in [q]^2.
\end{align*}
Since $\DTV{\pi_i}{\nu_i} \geq  n^{-1/4}$ for all $0 \leq i \leq m-1$, we have
\begin{align*}
	\mathbf{E}_{X \sim \pi}\left[f(X)\right] - \mathbf{E}_{Y\sim \nu}\left[f(Y)\right] \geq mn^{-1/4} \geq \frac{n}{20t} \cdot n^{-1/4} \geq 4 n^{2/3},
\end{align*}
where the last inequality holds because $t = O(\log n)$ and $n$ is sufficiently large. 
Given a sample $\tau \in ([q]^2)^{m}$, we say the event $A$ occurs if
\begin{align*}
	f(\tau) > \frac{\mathbf{E}_{X \sim \pi}\left[f(X)\right] + \mathbf{E}_{Y\sim \nu}\left[f(Y)\right]}{2}.
\end{align*}  
Note that $m \leq n$. By Hoeffding's inequality, it holds that 
\begin{align*}
	\Pr_{\pi}[A] &= 1 - \Pr_{\pi}[\overline{A}] \geq 1 - \exp\left( -\frac{2n^{4/3}}{m} \right) = 1 - o(1),\\
	\Pr_{\nu}[A] &\leq  \exp\left( -\frac{2n^{4/3}}{m} \right)= o(1).
\end{align*}
By the definition of total variation distance, it holds that
\begin{align}\label{eq:logn-lower-bound-dtv-bound}
	\DTV{\mathcal{X}}{\mathcal{Y}} = \DTV{\pi}{\nu} = \max_{B \subseteq ([q]^2)^m}|\pi(B) - \nu(B)| \geq \Pr_{\pi}[A] - \Pr_{\nu}[A] \geq 1 - o(1).
\end{align}

%
%

Recall that the above $\mathcal{Y}$ is sampled conditioning on an arbitrary configuration $\sigma_F\in[q]^F$ of fixed centers. Now we consider a $\sigma\in[q]^P$ sampled from the Gibbs distribution $\mu$ on the path $P$ and its restrictions $\sigma_F$, $\sigma_U$ and $\sigma_C$ on $F=\{x_i\}$, $U=\{u_i,v_i\}$ and $C=F\cup U$. Also let $\boldsymbol{X}$ be the vector of values returned by the vertices in $P$ in a $t$-round protocol, and $\boldsymbol{X}_F$, $\boldsymbol{X}_U$ and $\boldsymbol{X}_C$ its restrictions on the respective sets of centers. The theorem follows if we can show that $\DTV{\boldsymbol{X}}{\sigma}>\frac{1}{3}$ for our choice of $t=O(\log n)$.
By definition of the total variation distance, we have:
\begin{align}
&\quad\DTV{\boldsymbol{X}}{\sigma}
\ge
\DTV{\boldsymbol{X}_C}{\sigma_C}\notag\\
&=\frac{1}{2}\sum_{\sigma_F\in[q]^F}\sum_{\sigma_U\in[q]^U}\bigg(\big|\mu(\sigma_F,\sigma_U)-\Pr[\boldsymbol{X}_F=\sigma_F\wedge \boldsymbol{X}_U=\sigma_U]\big|\bigg)\notag\\
&=\frac{1}{2}\sum_{\sigma_F\in[q]^F}\sum_{\sigma_U\in[q]^U}\bigg(\big|\mu(\sigma_F)\mu(\sigma_U\mid \sigma_F)- \Pr[\boldsymbol{X}_F=\sigma_F]\Pr[ \boldsymbol{X}_U=\sigma_U]\big|\bigg)\notag\\
&\ge
\sum_{\sigma_F\in[q]^F}\mu(\sigma_F)\cdot \frac{1}{2}\sum_{\sigma_U\in[q]^U}\left|\mu(\sigma_U\mid \sigma_F) - \Pr[ \boldsymbol{X}_U=\sigma_U]\right|-\frac{1}{2}\sum_{\sigma_F\in[q]^F}\left|\mu(\sigma_F)-\Pr[\boldsymbol{X}_F=\sigma_F]\right|.\label{eq:logn-lower-bound-total}
\end{align}
Note that
\begin{align*}
\DTV{\boldsymbol{X}}{\sigma}\ge\DTV{\boldsymbol{X}_F}{\sigma_F}=\frac{1}{2}\sum_{\sigma_F\in[q]^F}\left|\mu(\sigma_F)-\Pr[\boldsymbol{X}_F=\sigma_F]\right|.
\end{align*}

If this quantity is greater than $1/3$, then we already have $\DTV{\boldsymbol{X}}{\sigma}>1/3$ and the lower bound is proved. If otherwise, we suppose that 
\[
\frac{1}{2}\sum_{\sigma_F\in[q]^F}\left|\mu(\sigma_F)-\Pr[\boldsymbol{X}_F=\sigma_F]\right|\le \frac{1}{3}.
\] 
Observe that for any $\sigma_F\in[q]^F$, we have
\begin{align*}
\frac{1}{2}\sum_{\sigma_U\in[q]^U}\left|\mu(\sigma_U\mid \sigma_F) - \Pr[ \boldsymbol{X}_U=\sigma_U]\right|&=\DTV{\mathcal{X}}{\mathcal{Y}}\ge 1-o(1),
\end{align*}
where $\mathcal{Y}=(\mathcal{Y}_i=(\sigma_{u_i},\sigma_{v_i}))_{0\le i\le m-1}$ is sampled conditioning on $\sigma_F$ and the inequality is due to~\eqref{eq:logn-lower-bound-dtv-bound}.

Therefore, the total variation distance in~\eqref{eq:logn-lower-bound-total} can be further bounded as
\begin{align*}
\DTV{\boldsymbol{X}}{\sigma}
\ge
\sum_{\sigma_F\in[q]^F}\mu(\sigma_F)(1-o(1))-\frac{1}{3}
=1-o(1)-\frac{1}{3}>\frac{1}{3}. &\qedhere
\end{align*}
\end{proof}

Next, we state a strong $\Omega(\Diam)$ lower bound for sampling with long-range correlations.

\subsection{An $\Omega(\Diam)$ lower bound in the non-uniqueness regime}

We consider the weighted independent sets of graphs, the \concept{hardcore model}. Given a graph $G(V,E)$ and a \concept{fugacity} parameter $\lambda>0$, each configuration $\sigma$ in
\[
\mathrm{IS}(G)=\left\{\sigma\in\{0,1\}^V:\forall(u,v)\in E,\sigma_u\sigma_v=0\right\}
\] 
indicates an independent set $I$ in $G$ and is assigned a weight $w(\sigma)=\lambda^{|I|}$. The Gibbs distribution $\mu=\mu_{G}$ is defined over all independent sets in $G$ proportional to their weights. As discussed in Section~\ref{sec:model-MRF}, the model is an MRF.

The hardcore model on graphs with maximum degree $\Delta$ undergoes a computational phase transition at the \concept{uniqueness threshold} $\lambda_c(\Delta)=\frac{(\Delta-1)^{\Delta-1}}{(\Delta-2)^\Delta}$, such that sampling from the Gibbs distribution can be done in polynomial time in the \concept{uniqueness regime} $\lambda<\lambda_c$~\cite{weitz2006counting, efthymiou2016convergence} and is intractable unless NP=RP in the \concept{non-uniqueness regime} $\lambda>\lambda_c$~\cite{Cai:2016:BSB:2912239.2912263,sly2010computational, sly2014counting, galanis2016inapproximability}.

The following theorem states an $\Omega(\Diam)$ lower bound for sampling from the hardcore model in the non-uniqueness regime.
In particular when $\lambda=1$ the model represents the uniform independent sets and the non-uniqueness $\lambda>\lambda_c(\Delta)$ holds when $\Delta\ge 6$, which gives us Theorem~\ref{main-theorem-diameter-lower-bound}.

\begin{theorem}\label{thm:hardcore-lowerbound}
Let $\Delta\ge3$ and $\lambda>\lambda_c(\Delta)$.
Let $\epsilon>0$ be a sufficiently small constant.
For all $N>0$ there exists a graph $\mathcal{G}$ on $\Theta(N)$ vertices with maximum degree $\Delta$ and diameter $\Diam(\mathcal{G})=\Omega(N^{1/11})$ such that for the hardcore model on $\mathcal{G}$ with fugacity $\lambda$,
any $t$-round protocol that samples within total variation distance $\epsilon$ from the Gibbs distribution $\mu=\mu_\mathcal{G}$ must have $t=\Omega(\Diam(\mathcal{G}))$.
\end{theorem}
We follow the approaches in~\cite{Cai:2016:BSB:2912239.2912263,sly2010computational, sly2014counting, galanis2015inapproximability, galanis2016inapproximability} for the computational phase transition. The network $\mathcal{G}=H^G$ is constructed by lifting a graph $H$ with a gadget $G$, such that sampling from the hardcore model on $H^G$ with $\lambda>\lambda_c(\Delta)$ effectively samples a maximum cut in $H$. We choose $H$ to be an even cycle, in which the maximum cut imposes a long-range correlation among vertices. And to sample with such a long-range correlation, the sampling algorithm must not be local.

Unlike the results of~\cite{Cai:2016:BSB:2912239.2912263,sly2010computational, sly2014counting, galanis2015inapproximability, galanis2016inapproximability} which are for computational complexity of approximate counting, here we prove unconditional lower bounds for sampling in the \LOCAL{} model. 
Our lower bound is due to the long-range correlations in the random max-cut rather than the computational complexity of optimization.
Technical-wise, this means that in addition to show that a max-cut in $H$ is sampled, we also need that the sampled  max-cut is distributed almost uniformly.

\subsubsection{The random graph gadget}
We now describe the random graph gadget which is essential to the hardness of sampling. The gadget is constructed in two steps. For positive integers $n,r$ and $\Delta$, we first describe the construction of the random  bipartite (multi)graph $\mathcal{G}^r_n$:
\begin{itemize}
\item Let $V^+$ and $V^-$ be two vertex sets with $|V^+|=|V^-|=n+r$, such that $V^\pm=U^\pm\uplus W^\pm$ where $\left\lvert U^\pm\right\rvert=n$ and $\left\lvert W^\pm\right\rvert=r$. Let $V=V^+\cup V^-$, $W=W^+\cup W^-$ and $U=U^+\cup U^-$.
\item 
Uniformly and independently sample $\Delta-1$ perfect matchings between $V^+$ and $V^-$ and then uniformly and independently sample a perfect matching between $U^+$ and $U^-$. The union of all these matchings gives us the random bipartite (multi)graph $\mathcal{G}^r_n$, in which every vertex in $U$ has degree $\Delta$ and every vertex in $W$ has degree $\Delta-1$.
\end{itemize}
Now we describe the second part of the construction. Let $0<\theta<\psi<1/8$ be constants. Let $r':=(\Delta-1)^{\lfloor\theta\log_{\Delta-1}{n}\rfloor+2\lfloor\frac{\psi}{2}\log_{\Delta-1}{n}\rfloor}$. Note that $r'=o(n^{1/4})$. First, we sample $G$ from the distribution $\mathcal{G}_n^{r'}$. Next, attach $k$ disjoint $(\Delta-1)$-ary trees of even depth $l$ (with $k=(\Delta-1)^{\lfloor\theta\log_{\Delta-1}{n}\rfloor}$ and $l=2\lfloor\frac{\psi}{2}\log_{\Delta-1}{n}\rfloor$) to $W^\pm$, such that every vertex in $W$ is a leaf of exactly one tree and the trees do not share common vertices with the bipartite graph $G$, apart from the vertices in $W$. Let $T^\pm$ denote the roots of those trees ($|T^+|=|T^-|=k$), called ``terminals''. We denote the family of graphs that can be constructed this way by $\tilde{\mathcal{G}}(k,n,\Delta)$. Note that our construction is still bipartite with size $\Theta(n)$ and the terminals in $T^+$ and $T^-$ belongs to distinct partitions of the bipartite graph.

The \concept{phase} of a configuration $\sigma$, denoted as $Y(\sigma)$, is defined as
\begin{align*}
Y(\sigma):=
\begin{cases}
+ & \text{if $\sum_{v\in U^+}{\sigma_v\ge\sum_{v\in U^-}{\sigma_v}}$,}\\
- & \text{if $\sum_{v\in U^+}{\sigma_v<\sum_{v\in U^-}{\sigma_v}}$.}
\end{cases}
\end{align*}

It is easy to verify that the random bipartite graph $\mathcal{G}^r_n$ in the first step is an expander with high probability. 
The following proposition was proved in~\cite{Cai:2016:BSB:2912239.2912263}.
\begin{proposition}[Lemma 8 \& Lemma 9 in~\cite{Cai:2016:BSB:2912239.2912263}]\label{prop:hardcore-non-unique}
If $\lambda>\lambda_c(\Delta)=\frac{(\Delta-1)^{\Delta-1}}{(\Delta-2)^\Delta}$ then there exist two constants $0<q^-<q^+<1$ such that the followings hold. 
Let $Q^\pm_T$ denote the product measure on configurations in $\{0,1\}^T$ so that the spin states are i.i.d.~Bernoulli with probability $q^\pm$ on $T^+$ and $q^\mp$ on $T^-$, that is:
\begin{align*}
Q^\pm_T(\sigma_T)=
&\left(q^\pm\right)^{\sum_{v\in T^+}{\sigma_v}}\left(1-q^\pm\right)^{|T^+|-\sum_{v\in T^+}{\sigma_v}}\\
&\cdot\left(q^\mp\right)^{\sum_{v\in T^-}{\sigma_v}}\left(1-q^\mp\right)^{|T^-|-\sum_{v\in T^-}{\sigma_v}}.
\end{align*}
For any $\delta>0$, there exists sufficiently large constant $N_0(\delta)$ such that for all $n>N_0(\delta)$ the followings hold  altogether with positive probability for $G\sim\tilde{\mathcal{G}}(k,n,\Delta)$: 
\begin{itemize}
\item (expander) $G$ is connected with ${\Diam}\left(G\right)=O(\log n)$;
\item (balanced phases) $\Pr_G\left[Y(\sigma)=\pm\right]\in[(1-\delta)/2,(1+\delta)/2]$;
\item (phase-correlated almost independence) $\forall\tau_T\in\{0,1\}^T$,
\[\Pr_G\left[\sigma_T=\tau_T \mid Y(\sigma)=\pm\right]/Q_T^\pm(\tau_T)\in[1-\delta,1+\delta];\]
\end{itemize}
where $\Pr_G$ is the probability law for $\sigma$ sampled from $\mu_{G}$.

By the probabilistic method, there exists a $G$ satisfying the above conditions.
\end{proposition}

\subsubsection{Reduction from Max-Cut}\label{sec:reduction-max-cut}

Let $H$ be a cycle with $m$ vertices where $m>0$ is an even integer. Fix constants $\theta=\psi=1/9$ and
let $G\in\tilde{\mathcal{G}}(2k,n,\Delta)$, 
with $k=\Theta(m^{10/9})$ and $n=\Theta(k^{1/\theta})=\Theta(m^{10})$, be the graph that satisfies the conditions in Proposition~\ref{prop:hardcore-non-unique}.

\begin{itemize}
\item For each vertex $x\in H$ let $G_x$ be a copy of $G$. We denote by $T^\pm_x$ the respective set of $2k$ terminals in $G_x$. Let $\widehat{H}^G$ be the disconnected copies of the $G_x$, $x\in H$.
\item For every edge $(x,y)\in H$, add $k$ edges between $T^+_x$ and $T^+_y$ and similarly add $k$ edges between $T^-_x$ and $T^-_y$. This can be done in such a way that the resulting (multi)graph $H^G$ is $\Delta$-regular.
\end{itemize}
\begin{definition}\label{phase}
For each $x\in H$, we write $Y_x=Y_x(\sigma)$ for the phase of a configuration $\sigma$ on $G_x$. Let $\mathcal{Y}=(Y_x)_{x\in H}\in\{+,-\}^{V(H)}$. Given the phase $\mathcal{Y}'\in\{+,-\}^{V(H)}$, we define:
\begin{align*}
Z_{H^G}(\mathcal{Y}')=\sum\limits_{\sigma\in \mathrm{IS}(H^G)}{\lambda^{\|\sigma\|_1}\mathbf{1}\{\mathcal{Y}(\sigma)=\mathcal{Y}'\}},
\end{align*}
where
\begin{align*}
\mathrm{IS}(H^G)=\left\{\sigma\in\{0,1\}^{V(H^G)}: \forall uv\in E(H^G), \sigma_u\sigma_v=0\right\}
\end{align*}
is the set of all independent sets in~$H^G$. We also use $\Pr_{H^G}$ to represent the probability law for $\sigma$ sampled from $\mu_{H^G}$.
\end{definition}

Note that the cycle $H$ has precisely two maximum cuts. A key property for proving the lower bound is that in the non-uniqueness regime, sampling from the hardcore model on graph $H^G$ corresponds to sampling a maximum cut in $H$ almost uniformly. 

\begin{theorem}\label{thm:maxcut-hardcore}
Let $\lambda>\lambda_c(\Delta)$. Let $\mathcal{Y}_1,\mathcal{Y}_2\in\{+,-\}^{V(H)}$ correspond respectively to the two maximum cuts in $H$. It holds that:
\begin{align}
\Pr\nolimits_{H^G}\left[\mathcal{Y}(\sigma)=\mathcal{Y}_1\right]=\Pr\nolimits_{H^G}\left[\mathcal{Y}(\sigma)=\mathcal{Y}_2\right]\ge\frac{1}{2}-o(1).\label{eq:maxcut2-hardcore}
\end{align}
\end{theorem}

The theorem is implied by the following lemma, which is proved by applying a calculation in~\cite{sly2010computational} with the improved gadget property Proposition~\ref{prop:hardcore-non-unique}. 
\begin{lemma}\label{lem:cut-hardcore}
Let $\mathcal{Y}',\mathcal{Y}''\in\{+,-\}^{V(H)}$ and $\delta>0$. Suppose that $G$ satisfies the conditions in Proposition~\ref{prop:hardcore-non-unique}. It holds that
\begin{align*}
\frac{\Pr_{H^G}\left[\mathcal{Y}(\sigma)=\mathcal{Y}'\right]}{\Pr_{H^G}\left[\mathcal{Y}(\sigma)=\mathcal{Y}''\right]}\ge\left(\frac{1-\delta}{1+\delta}\right)^{2m}(\Theta/\Gamma)^{k[Cut(\mathcal{Y}')-Cut(\mathcal{Y}'')]},
\end{align*}
where $\Theta=(1-q^+q^-)^2$ and $\Gamma=(1-(q^+)^2)(1-(q^-)^2)$; and $Cut(\mathcal{Y})=|\{(x,y)\in E(H):\mathcal{Y}_x\ne\mathcal{Y}_y\}|$ for a $\mathcal{Y}\in\{+,-\}^{V(H)}$.
\end{lemma}

\begin{proof}
Since the graph $\widehat{H}^G$ consists of a collection of disconnected copies of $G$, the distribution of a configuration on $\widehat{H}^G$ is given by the product measure of configurations on the $(G_x)_{x\in H}$. In particular the phases are independent, therefore
\begin{align}
&\quad\frac{Z_{\widehat{H}^G}(\mathcal{Y}')}{Z_{\widehat{H}^G}(\mathcal{Y}'')}=\frac{Z_{\widehat{H}^G}(\mathcal{Y}')/Z_{\widehat{H}^G}}{Z_{\widehat{H}^G}(\mathcal{Y}'')/Z_{\widehat{H}^G}} =\frac{\Pr_G\left[Y(\sigma)=+\right]^{\sum\limits_{x\in H}{\mathbf{1}\{Y'_x=+\}}}\cdot\Pr_G\left[Y(\sigma)=-\right]^{\sum\limits_{x\in H}{\mathbf{1}\{Y'_x=-\}}}}{\Pr_G\left[Y(\sigma)=+\right]^{\sum\limits_{x\in H}{\mathbf{1}\{Y''_x=+\}}}\cdot\Pr_G\left[Y(\sigma)=-\right]^{\sum\limits_{x\in H}{\mathbf{1}\{Y''_x=-\}}}}
\notag\\
&\ge\left(\frac{1-\delta}{1+\delta}\right)^m. \label{eq:cut-hardcore1}
\end{align}
Note that the ratio $Z_{H^G}(\mathcal{Y}')/Z_{\widehat{H}^G}(\mathcal{Y}')$ is precisely the probability of a $\sigma$ sampled from $\mu_{\widehat{H}^G}$ being an independent set in $H^G$. And due to Proposition~\ref{prop:hardcore-non-unique},  conditioning on the phase $\mathcal{Y}'$ the spins of $\sigma_{\bigcup_{x\in H}{T_x}}$ are almost independent i.i.d.~Bernoulli with probabilities $q^+$ or $q^-$ depending on the phase, therefore
\begin{align}
\frac{Z_{H^G}(\mathcal{Y}')}{Z_{\widehat{H}^G}(\mathcal{Y}')}&=\mathrm{Pr}_{\widehat{H}^G}\left[\sigma\text{ is an IS in } H^G\mid\mathcal{Y}(\sigma)=\mathcal{Y}'\right]
\notag\\&=\mathrm{Pr}_{\widehat{H}^G}\left[\forall(u,v)\in E(H^G)\setminus E(\widehat{H}^G),\sigma_u\sigma_v\ne 1\mid\mathcal{Y}(\sigma)=\mathcal{Y}'\right]
\notag\\&\ge(1-\delta)^m\sum\limits_{\sigma_{\bigcup_{x\in H}{T_x}}}{Q_{\sigma_T}(\mathcal{Y}')}
\notag\\
&=(1-\delta)^m\Gamma^{k|E(H)|}(\Theta/\Gamma)^{kCut(\mathcal{Y}')},\label{eq:cut-hardcore2}
\end{align}
where
\begin{align*}Q_{\sigma_T}(\mathcal{Y}')=&\Bigg[\mathbf{1}\{\forall (u,v)\in E(H^G)\setminus E(\widehat{H}^G),\sigma_u\sigma_v\ne1\} \times\prod\limits_{x\in H}{Q_{T_x}^{Y'_x}(\sigma_{T_x})}\Bigg].
\end{align*}
Similarly, we can obtain
\begin{align}
\frac{Z_{H^G}(\mathcal{Y}'')}{Z_{\widehat{H}^G}(\mathcal{Y}'')}\le(1+\delta)^m\Gamma^{k|E(H)|}(\Theta/\Gamma)^{kCut(\mathcal{Y}'')}.\label{eq:cut-hardcore3}
\end{align}
Combining~\eqref{eq:cut-hardcore1}, \eqref{eq:cut-hardcore2} and \eqref{eq:cut-hardcore3}, we have:
\begin{align*}
\frac{\Pr_{H^G}\left[\mathcal{Y}(\sigma)=\mathcal{Y}'\right]}{\Pr_{H^G}\left[\mathcal{Y}(\sigma)=\mathcal{Y}''\right]}=\frac{Z_{H^G}(\mathcal{Y}')}{Z_{H^G}(\mathcal{Y}'')}
&\ge\left(\frac{1-\delta}{1+\delta}\right)^m(\Theta/\Gamma)^{k[Cut(\mathcal{Y}')-Cut(\mathcal{Y}'')]}\cdot\frac{Z_{\widehat{H}^G}(\mathcal{Y}')}{Z_{\widehat{H}^G}(\mathcal{Y}'')}\\
&\ge\left(\frac{1-\delta}{1+\delta}\right)^{2m}(\Theta/\Gamma)^{k[Cut(\mathcal{Y}')-Cut(\mathcal{Y}'')]}. \qedhere
\end{align*}
\end{proof}

\begin{proof}[Proof of Theorem~\ref{thm:maxcut-hardcore}:]
Let $\mathcal{Y}',\mathcal{Y}''\in\{+,-\}^{V(H)}$ such that $Cut(\mathcal{Y}')>Cut(\mathcal{Y}'')$. Let $\delta>0$, by Lemma \ref{lem:cut-hardcore}, we have
\begin{align*}
\frac{\Pr_{H^G}\left[\mathcal{Y}(\sigma)=\mathcal{Y}'\right]}{\Pr_{H^G}\left[\mathcal{Y}(\sigma)=\mathcal{Y}''\right]}&\ge\left(\frac{1-\delta}{1+\delta}\right)^{2m}(\Theta/\Gamma)^{k[Cut(\mathcal{Y}')-Cut(\mathcal{Y}'')]}.
\end{align*}
Note that for $\lambda>\lambda_c(\Delta)=\frac{(\Delta-1)^{\Delta-1}}{(\Delta-2)^\Delta}$, we have $\Theta>\Gamma$. Thus for $k=\Theta(m^{10/9})$  we have
\begin{align*}
\frac{\Pr_{H^G}\left[\mathcal{Y}(\sigma)=\mathcal{Y}'\right]}{\Pr_{H^G}\left[\mathcal{Y}(\sigma)=\mathcal{Y}''\right]}\ge\left(\frac{1-\delta}{1+\delta}\right)^{2m}(\Theta/\Gamma)^k\ge4^m.
\end{align*}
Since the size of $\{+,-\}^{V(H)}$ is at most $2^m$, it follows that with probability at least $1-o(1)$ the phases $\mathcal{Y}(\sigma)$ attain a maximum cut in $H$. Therefore, we only need to prove $Z_{H^G}(\mathcal{Y}_1)=Z_{H^G}(\mathcal{Y}_2)$ for the two maximum cuts $\mathcal{Y}_1$ and $\mathcal{Y}_2$ in $H$. By simple calculation, we have
\begin{align*}
Z_{H^G}(\mathcal{Y}_1)&=Z_{\widehat{H}^G}(\mathcal{Y}_1)\cdot\Pr\nolimits_{\widehat{H}^G}\left[\sigma\in \mathrm{IS}(H^G)\mid \mathcal{Y}(\sigma)=\mathcal{Y}_1\right]
\\&=Z_{\widehat{H}^G}(\mathcal{Y}_1)\cdot\Pr\nolimits_{\widehat{H}^G}\left[\forall (u,v)\in E(H^G)\setminus E(\widehat{H}^G),\sigma_u\sigma_v\ne 1\mid\mathcal{Y}(\sigma)=\mathcal{Y}_1\right]
\\&=Z_{\widehat{H}^G}\cdot\Pr\nolimits_G\left[Y=+\right]^{m/2}\cdot\Pr\nolimits_G\left[Y=-\right]^{m/2}
\\&\quad\cdot\Pr\nolimits_{\widehat{H}^G}\left[\forall (u,v)\in E(H^G)\setminus E(\widehat{H}^G),\sigma_u\sigma_v\ne 1\mid\mathcal{Y}(\sigma)=\mathcal{Y}_1\right]
\end{align*}
and
\begin{align*}
 Z_{H^G}(\mathcal{Y}_2)&=Z_{\widehat{H}^G}(\mathcal{Y}_2)\cdot\Pr\nolimits_{\widehat{H}^G}\left[\sigma\in \mathrm{IS}(H^G)\mid \mathcal{Y}(\sigma)=\mathcal{Y}_2\right]
\\&=Z_{\widehat{H}^G}(\mathcal{Y}_2)\cdot\Pr\nolimits_{\widehat{H}^G}\left[\forall (u,v)\in E(H^G)\setminus E(\widehat{H}^G),\sigma_u\sigma_v\ne 1\mid\mathcal{Y}(\sigma)=\mathcal{Y}_2\right]
\\&=Z_{\widehat{H}^G}\cdot\Pr\nolimits_G\left[Y=+\right]^{m/2}\cdot\Pr\nolimits_G\left[Y=-\right]^{m/2}
\\&\quad \cdot\Pr\nolimits_{\widehat{H}^G}\left[\forall (u,v)\in E(H^G)\setminus E(\widehat{H}^G),\sigma_u\sigma_v\ne 1\mid\mathcal{Y}(\sigma)=\mathcal{Y}_2\right].
\end{align*}
By symmetry of the even-length cycle, it holds that
\begin{align*}
&\quad\Pr\nolimits_{\widehat{H}^G}\left[\forall(u,v)\in E(H^G)\setminus E(\widehat{H}^G),\sigma_u\sigma_v\ne 1\mid\mathcal{Y}(\sigma)=\mathcal{Y}_1\right]\\&=\Pr\nolimits_{\widehat{H}^G}\left[\forall(u,v)\in E(H^G)\setminus E(\widehat{H}^G),\sigma_u\sigma_v\ne 1\mid\mathcal{Y}(\sigma)=\mathcal{Y}_2\right]. \qedhere
\end{align*} 
\end{proof}

\subsubsection{Proof of the $\Omega(\Diam)$ lower bound}
Now we are ready to prove Theorem~\ref{thm:hardcore-lowerbound}.
Let $N$ be sufficiently large. We choose an integer $n=\Theta(N^{10/11})$ and even integer $m=\Theta(N^{1/11})$ such that $m/2$ is odd, so that 
a gadget $G$ is constructed to satisfy Proposition~\ref{prop:hardcore-non-unique}, and the graph $\mathcal{G}=H^G$, where $H$ is a cycle of length $m$, is constructed as described in Section~\ref{sec:reduction-max-cut}.
Note that $\text{diam}\left(\mathcal{G}\right)\ge\Diam(H)\ge m/2$ and $\left\lvert V\left(\mathcal{G}\right)\right\rvert=\Theta(N)$, therefore $\text{diam}\left(\mathcal{G}\right)=\Omega(N^{1/11})$.

Let $\sigma'$ denote the output of a $t$-round protocol with $t\le 0.49\cdot \Diam(\mathcal{G})$ on network $\mathcal{G}$, whose distribution is denoted as $\mu_t$; and let $\sigma$ be sampled from the hardcore Gibbs distribution $\mu=\mu_{\mathcal{G}}$. 
By contradiction, we assume that $\DTV{\mu_t}{\mu}\le \epsilon$ for sufficiently small constant $\epsilon$.

Let $\mathcal{Y}',\mathcal{Y}''\in\{+,-\}^{V(H)}$ denote the phases corresponding to the two maximum cuts in the cycle $H$. Therefore, by Theorem~\ref{thm:maxcut-hardcore}, we have
\[
\Pr[\mathcal{Y}(\sigma)\in\{\mathcal{Y}',\mathcal{Y}''\}]\ge1-o(1).
\]
We pick $u,v\in V(\mathcal{G})$ which satisfy that $\mathrm{dist}_{\mathcal{G}}(u,v)=\Diam\left(\mathcal{G}\right)$. Since $\mathcal{G}=H^G$ is constructed by replacing each vertex $x$ in $H$ with $G_x$ which is an identical copy of $G$, 
it must hold that $u\in G_x, v\in G_y$ for some vertices $x,y$ in $H$ with $\dist_H(x,y)=m/2$. And since $m/2$ is odd, without loss of generality, we suppose that  $Y'_x=+,Y'_y=-$ and $Y''_x=-, Y''_y=+$. 
Moreover, for all $u'\in G_x, v'\in G_y$, by the triangle inequality we have:
\begin{align*}
\mathrm{dist}_{\mathcal{G}}(u,u')+\mathrm{dist}_{\mathcal{G}}(u',v')+\mathrm{dist}_{\mathcal{G}}(v',v)&\ge \mathrm{dist}_{\mathcal{G}}(u,v)=\text{diam}\left(\mathcal{G}\right).
\end{align*}
Due to Proposition~\ref{prop:hardcore-non-unique}, it holds that  $\Diam(G)=O(\log n)$, thus we have:
\[
\dist_{\mathcal{G}}(u',v')\ge\Diam(\mathcal{G})-O(\log n)=(1-o(1))\Diam\left(\mathcal{G}\right).
\] 
For the $\sigma'$ returned by a $t$-round protocol where $t\le 0.49\cdot \Diam(\mathcal{G})$,  according to the property~\eqref{eq:non-local-independence},
the $\sigma'_{G_x}$ and $\sigma'_{G_y}$ are independent of each other, thus the phases of $G_x$ and $G_y$ on $\sigma'$ are independent of each other:
\begin{align}
\Pr\left[Y_x(\sigma')=+\mid Y_y(\sigma')=-\right]=\Pr\left[Y_x(\sigma')=+\mid Y_y(\sigma')=+\right].\label{eq:lowerbound-hardcore1}
\end{align}
On the other hand, since $\DTV{\sigma'}{\sigma}\le \epsilon$, we have
\begin{align*}
\Pr\left[Y_x(\sigma')=+\mid Y_y(\sigma')=-\right]&=\frac{\Pr\left[Y_x(\sigma')=+\wedge Y_y(\sigma')=-\right]}{\Pr\left[Y_y(\sigma')=-\right]}
\\&\ge\frac{\Pr\left[Y_x(\sigma)=+\wedge Y_y(\sigma)=-\right]-\epsilon}{\Pr\left[Y_y(\sigma)=-\right]+\epsilon}\quad\text{(by $d_{\text{TV}}(\sigma',\sigma)\le\epsilon$)}
\\&\ge\frac{\Pr\left[\mathcal{Y}(\sigma)=\mathcal{Y}'\right]-\epsilon}{\Pr\left[Y_y(\sigma)=-\right]+\epsilon}
\\&\ge\frac{1/2-o(1)-\epsilon}{\Pr\left[\mathcal{Y}(\sigma)\ne\mathcal{Y}''\right]+\epsilon}\ge\frac{1-2\epsilon-o(1)}{1+2\epsilon+o(1)}, \quad\text{(by Theorem \ref{thm:maxcut-hardcore})}
\end{align*}
and
\begin{align*}
\Pr\left[Y_x(\sigma')=+\mid Y_y(\sigma')=+\right] &=\frac{\Pr\left[Y_x(\sigma')=+\wedge Y_y(\sigma')=+\right]}{\Pr\left[Y_y(\sigma')=+\right]}
\\&\le\frac{\Pr\left[Y_x(\sigma)=+\wedge Y_y(\sigma)=+\right]+\epsilon}{\Pr\left[Y_y(\sigma)=+\right]-\epsilon}\quad\text{(by $d_{\text{TV}}(\sigma',\sigma)\le\epsilon$)}
\\&\le\frac{\Pr\left[\mathcal{Y}(\sigma)\notin\left\{\mathcal{Y}',\mathcal{Y}''\right\}\right]+\epsilon}{\Pr\left[\mathcal{Y}(\sigma)=\mathcal{Y}''\right]-\epsilon}
\\&\le\frac{2\epsilon+o(1)}{1-2\epsilon-o(1)}.\qquad\qquad\qquad\qquad\qquad\text{(by Theorem \ref{thm:maxcut-hardcore})}
\end{align*}
This implies that
\begin{align*}
\Pr\left[Y_x(\sigma')=+\mid Y_y(\sigma')=+\right]<\Pr\left[Y_x(\sigma')=+\mid Y_y(\sigma')=-\right]
\end{align*}
by taking $\epsilon$ to be a sufficiently small constant, which contradicts the independence given in~\eqref{eq:lowerbound-hardcore1}.

\section{Conclusion}
In this paper, we study the local sampling problem and ask a new question about local computation:  \emph{whether a locally definable joint distribution can be sampled locally.} 

On the positive side, we give two distributed sampling algorithms \LubyGlauber{} and \LocalMetropolis{}. We show that \LubyGlauber{} achieves $O(\Delta \log n)$ mixing time under Dobrushin's condition and \LocalMetropolis{} may achieve optimal $O(\log n)$ mixing time under a stronger mixing condition. Thus many locally definable joint distributions can be sampled locally.

On the negative side, we give an $\Omega(\log n)$ lower bound for sampling from a broad class of locally defined joint distributions. Thus the $O(\log n)$-radius can be considered as the new criteria for being local for distributed sampling algorithms. Furthermore, we give an $\Omega(\mathrm{diam})=n^{\Omega(1)}$ lower bound for sampling weighted independent sets in the non-uniqueness regime. Since independent set is trivial to construct, this gives a strong separation between local sampling and local construction. The lower bounds hold even if every vertex is aware of the graph structure, which means the hardness for local sampling is due to the discrepancy between the locality of randomness in distributed algorithms and the long-range correlation in the joint distribution from which we want to sample.

\section*{Acknowledgements}
Yitong Yin wants to thank Daniel {\v{S}}tefankovi{\v{c}} for the stimulating discussions in the beginning of this project. He also wants to thank Heng Guo, Tom Hayes, Eric Vigoda, and Chaodong Zheng for helpful discussions.

\end{document}